\documentclass[sigconf,nonacm]{acmart}

\usepackage[frozencache,cachedir=minted-cache]{minted} 
\usepackage{amsmath,amsfonts}
\usepackage{graphicx}
\usepackage{textcomp}
\usepackage{algorithm}
\usepackage{algpseudocode}
\usepackage{amsthm}
\usepackage{nicefrac}
\usepackage{enumitem}

\setminted[python]{breaklines, framesep=2mm, fontsize=\footnotesize, numbersep=5pt}

\usepackage{subcaption}
\usepackage{xcolor}

\DeclareMathOperator*{\argmaxA}{arg\,max} 

\AtBeginDocument{%
  \providecommand\BibTeX{{%
    \normalfont B\kern-0.5em{\scshape i\kern-0.25em b}\kern-0.8em\TeX}}}

\setcopyright{acmcopyright}
\copyrightyear{2023}
\acmYear{2023}
\acmDOI{XXXXXXX.XXXXXXX}

\begin{document}

\title{A Framework to Prioritize Software Quality Practices for Machine Learning Systems}
\title{Best Practices for Machine Learning Systems: An Industrial Framework for Analysis and Optimization}

\author{Georgios Christos Chouliaras}
\affiliation{%
  \institution{Booking.com}
  \city{Amsterdam}
  \country{The Netherlands}
  \postcode{1017 CE}
}
\email{georgios.chouliaras@booking.com}

\author{Kornel Kiełczewski}
\affiliation{%
  \institution{Booking.com}
  \city{Amsterdam}
  \country{The Netherlands}
}
\email{kornel.kielczewski@booking.com}

\author{Amit Beka}
\affiliation{%
  \institution{Booking.com}
  \city{Tel Aviv}
  \country{Israel}
}
\email{amit.beka@booking.com}

\author{David Konopnicki}
\affiliation{%
  \institution{Booking.com}
  \city{Tel Aviv}
  \country{Israel}
}
\email{david.konopnicki@booking.com}

\author{Lucas Bernardi}
\affiliation{%
  \institution{Booking.com}
  \city{Amsterdam}
  \country{The Netherlands}
}
\email{lucejb@gmail.com}

\begin{abstract}
 In the last few years, the Machine Learning (ML) and Artificial Intelligence community has developed an increasing interest in Software Engineering (SE) for ML Systems leading to a proliferation of best practices, rules, and guidelines aiming at improving the quality of the software of ML Systems. However, understanding their impact on the overall quality has received less attention. Practices are usually presented in a prescriptive manner, without an explicit connection to their overall contribution to software quality. Based on the observation that different practices influence different aspects of software-quality and that one single quality aspect might be addressed by several practices we propose a framework to analyse sets of best practices with focus on quality impact and prioritization of their implementation. We first introduce a hierarchical Software Quality Model (SQM) specifically tailored for ML Systems. Relying on expert knowledge, the connection between individual practices and software quality aspects is explicitly elicited for a large set of well-established practices. Applying set-function optimization techniques we can answer questions such as what is the set of practices that maximizes SQM coverage, what are the most important ones, which practices should be implemented in order to improve specific quality aspects, among others. We illustrate the usage of our framework by analyzing well-known sets of practices. 
\end{abstract}


\maketitle

\section{Introduction}
In Software Engineering, Software Quality Models (SQM) are central when it comes to achieving high quality software, as highlighted for example by \cite{discipline}: \textit{"A quality model provides the framework towards a definition of quality"}.
A Software Quality Model is the set of \textit{characteristics} and the
relationships between them that provides the basis for specifying quality requirements and evaluation \cite{ISO9126}. In practice, a SQM is a structured set of attributes describing the aspects that are believed contribute to the overall quality. Machine Learning (ML) systems have unique properties like data dependencies and hidden feedback loops which make quality attributes such as \textit{diversity}, \textit{fairness}, \textit{human agency} and \textit{oversight} more relevant than in traditional software systems \cite{Siebert2021}. This makes traditional quality models not directly applicable for ML applications. Moreover in recent years there has been a rise in the publication of best practices tailored for ML systems \cite{HiddenTechDebt}, \cite{Amershi}, \cite{Wujek2016BestPF}, \cite{Zhang_ml_testing}, \cite{Serban}, however understanding their impact on overall quality and the systematic prioritization for their adoption has not received enough interest. Improving the quality of ML systems, especially in an industrial setting where multiple ML systems are in production, does not only require a set of practices, but also a deep understanding of their contribution to specific aspects of the quality of the system, as well as criteria to prioritize their implementation due to their large number and high implementation costs. Without a systematic prioritization based on their contribution to each individual aspect of software quality, it is challenging for practitioners to choose the optimal practices to adopt based on their needs which might lead to limited adoption, undesired biases, inefficient development processes and inconsistent quality. The challenge lies on the fact that some best-practices have a narrow impact, strongly affecting a few specific quality aspects while others have wider impact affecting many aspects, which might lead to redundancy or gaps in the \textit{coverage} of the all the relevant quality aspects. Another challenge is that the importance of each quality aspect depends on the specific ML application, hence there is no single set of best-practices that satisfies the quality requirements of all ML applications. To address these challenges we introduce a reusable framework to analyse the contribution of a set of best practices to the quality of the system according to the specific needs of the particular application. The framework consists of a general-purpose Software Quality Model for ML Systems, expert-based representations of a large set of well established best-practices, and a criterion to assess a \textit{set} of best practices w.r.t. our SQM: the \textit{SQM Coverage Criterion}, which quantifies how many of the attributes receive enough attention from a given \textit{set} of best practices. Applying set optimization techniques we can answer questions such as what are the practices that maximize the coverage, which practices can be implemented to address specific quality aspects and which aspects lack coverage, among others.

Concretely, our contributions are the following: \textbf{1)} A general-purpose software quality model tailored for ML systems. \textbf{2)} A framework to analyse and prioritize software engineering best practices based on their influence on quality, with the flexibility to be adaptable according to the needs of each organization. \textbf{3)} We apply the proposed framework to analyze existing sets of best practices for ML systems and identify their strengths and potential gaps.

The rest of the paper is organized as follows. Section \ref{section:related} discusses related work with emphasis on Software Quality Models and software best-practices for ML systems, section \ref{section:sqm} introduces our Software Quality Model and describes its construction process. Section \ref{section:framework} introduces our best-practices analysis framework with details about its construction process and relevant algorithms. In section \ref{section:application} various best-practices sets are analysed using our framework, we present our findings and insights. Finally, section \ref{section:conclusion} summarizes our work and discuses limitations and future work. Appendices include all the details, such as proofs, extensive results, and computer code to facilitate reusability and repeatability of our framework.

\section{Related work}
\label{section:related}
\subsection{Software Quality Models for ML Systems}
Defining and measuring software quality is a fundamental problem and one of the first solutions came through the means of a software quality model in 1978 \cite{McCall}. Such models include general software \textit{characteristics} which are further refined into \textit{sub-characteristics}, which are decomposed into measurable software attributes whose values are computed by a metric \cite{Botella}. 

Software quality models developed until 2001 \cite{Boehm}, \cite{Grady}, \cite{Dromey}, \cite{ISO9126} are characterized as \textit{basic} since they make global assessments of a software product. Models developed afterwards, such as \cite{bertoa}, \cite{Alvaro}, \cite{ISO25012} are built on top of basic models and are specific to certain domains or specialized applications, hence are called \textit{tailored} quality models \cite{Miguel}. Such a quality model tailored for data products has been presented in \cite{ISO25012}.

Software for ML Systems exhibits differences when compared to traditional software such as the fact that minor changes in the input may lead to large discrepancies in the output \cite{Kurakin}. Moreover due to the dependencies on data, ML systems accumulate technical debt which is harder to recognize than code dependencies, which are identified via static analysis by compilers and linkers, tooling that is not widely available for data dependencies. Other peculiarities of ML systems include direct and hidden feedback loops where two systems influence each other indirectly \cite{HiddenTechDebt}. Additionally, software quality aspects such as \textit{fairness} and \textit{explainability} as well as legal and regulatory aspects which are relevant to ML software are not covered by existing software quality models \cite{Vogelsang}. Furthermore, existing quality attributes such as \textit{maintainability} and \textit{testability} need to be rethought in the context of ML software \cite{Horkoff}. All these peculiarities make existing software quality models only partially applicable to ML software. In \cite{Siebert2021} the authors present the systematic construction of quality models for ML systems based on a specific industrial use case. The authors focus on the process of constructing a quality meta model, identifying ML quality requirements based on the use case and instantiating a quality model that is tailored to the business application. In our work however, we introduce a general software quality model for ML systems that can be directly applied on a large set of industrial applications, without the need to go through a construction process. The key difference between our work and \cite{Siebert2021}, is that their main contribution a development process for quality models, while one of our main contributions is the quality model itself, which can be used with no or minimum modifications for a broad range of ML systems. This allows the usage of the same quality model for multiple use cases within an organization which reduces the effort of its adoption and allows to create a common communication language regarding the quality of the ML systems in the organization. In \cite{MYLLYAHO} the authors conclude that the majority of the studies on software quality for ML either adopt or extend the ISO 25010 Quality Model for software product quality \cite{iso25010}. They find though that there is no consensus on whether ISO 25010 is appropriate to use for AI-based software or which characteristics of AI-based software may be mapped to attributes of traditional quality models. Unlike other studies, we did not adopt or extend ISO 25010 but rather followed a systematic approach to build our quality model from scratch by adding quality sub-characteristics based on their relevance to ML systems.

\subsection{Software best-practices for ML Systems}
Best practices for increasing the quality of ML systems are presented in \cite{Rubric}, \cite{Amershi} and \cite{Wujek2016BestPF} however a systematic way to link the influence of the recommended practices to the software quality attributes of ML systems is not included. This makes it particularly challenging for ML practitioners to prioritize the adoption (or even understand the impact) of the large set of best practices based on the specific needs of their organizations. In \cite{Zhang_ml_testing} the authors present published ML practices targeting several testing properties (\textit{relevance}, \textit{robustness}, \textit{correctness}, \textit{efficiency}, \textit{security}, \textit{privacy}, \textit{fairness} and \textit{interpretability}) however their influence on quality aspects is not being studied. The authors in \cite{Serban} conducted a survey of ML practitioners from multiple companies and present the effect of various published ML practices on four categories (\textit{Agility}, \textit{Software Quality}, \textit{Team Effectiveness} and \textit{Traceability}). They present the importance of each practice for each of the categories, as perceived by the surveyed practitioners. However, these categories are generic, and in fact only two of them are directly related to software quality (\textit{Software Quality} and \textit{Traceability}), in contrast, we study the influence of each best practice on a full-blown general purpose Software Quality Model specifically built for ML system with fine-grained aspects such as \textit{testability} and \textit{deployability}. Furthermore, we study the influence on each quality aspect of the quality model when a \textit{set} of practices is applied, which is key to understand and prioritize best-practices since the overall impact is different depending on which other practices are also implemented. In \cite{LWAKATARE} the authors extracted challenges and solutions for large scale ML systems synthesized into four quality attributes: \textit{adaptability}, \textit{scalability}, \textit{safety} and \textit{privacy}. They categorized software practices based on the step on the ML lifecycle and the addressed quality attribute. A difference of this work with ours, is that in \cite{LWAKATARE} each practice targets a single quality attribute while its effect on multiple attributes is not explicitly studied. Even though there is work that studies the effect of practices on software quality \cite{Siebert2021}, \cite{LWAKATARE}, \cite{Serban} to the best of our knowledge, no study has been published about the interrelationship of software best-practices for ML Systems with multiple fine-grained quality attributes, nor about their prioritization in order to balance Software Quality and implementation costs.

\section{A software quality model for ML systems}
\label{section:sqm}

\subsection{The model}

A quality model determines which quality aspects are considered when evaluating the properties of a software product \cite{iso25010}. Our software quality model for ML systems comprises 7 quality \textit{characteristics} further divided into \textit{sub-characteristics}. Quality \textit{characteristics} are general properties of quality that comprise the fundamental factors, which cannot be measured directly. Each \textit{characteristic} consists of \textit{sub-characteristics}, which are concrete quality aspects that can be directly influenced and measured. A graphical illustration of our software quality model for ML systems is presented in tree-structure in Figure \ref{fig:qm}. We define quality \textit{characteristics} as follows:
\begin{description}
    \item \textbf{Utility} — The degree to which a machine learning system provides functions that meet stated and implied needs when used under specified conditions.
    \item \textbf{Economy} — The level of performance relative to the amount of resources used under stated conditions.
    \item \textbf{Robustness} — The tolerance to degradation by the machine learning system under consideration when exposed to dynamic or adverse events.
    \item \textbf{Modifiability} — The degree of effectiveness and efficiency with which a machine learning system can be modified to improve it, correct it or adapt it to changes in environment and in requirements.
    \item \textbf{Productionizability} — The ease of performing the actions required for a machine learning system to run successfully in production.
    \item \textbf{Comprehensibility} — The degree to which users and contributors understand the relevant aspects of a machine learning system.
    \item \textbf{Responsibility} — The level of trustworthiness of a machine learning system.
\end{description}
\begin{figure*}[htbp]
  \centering
  \includegraphics[width=397pt]{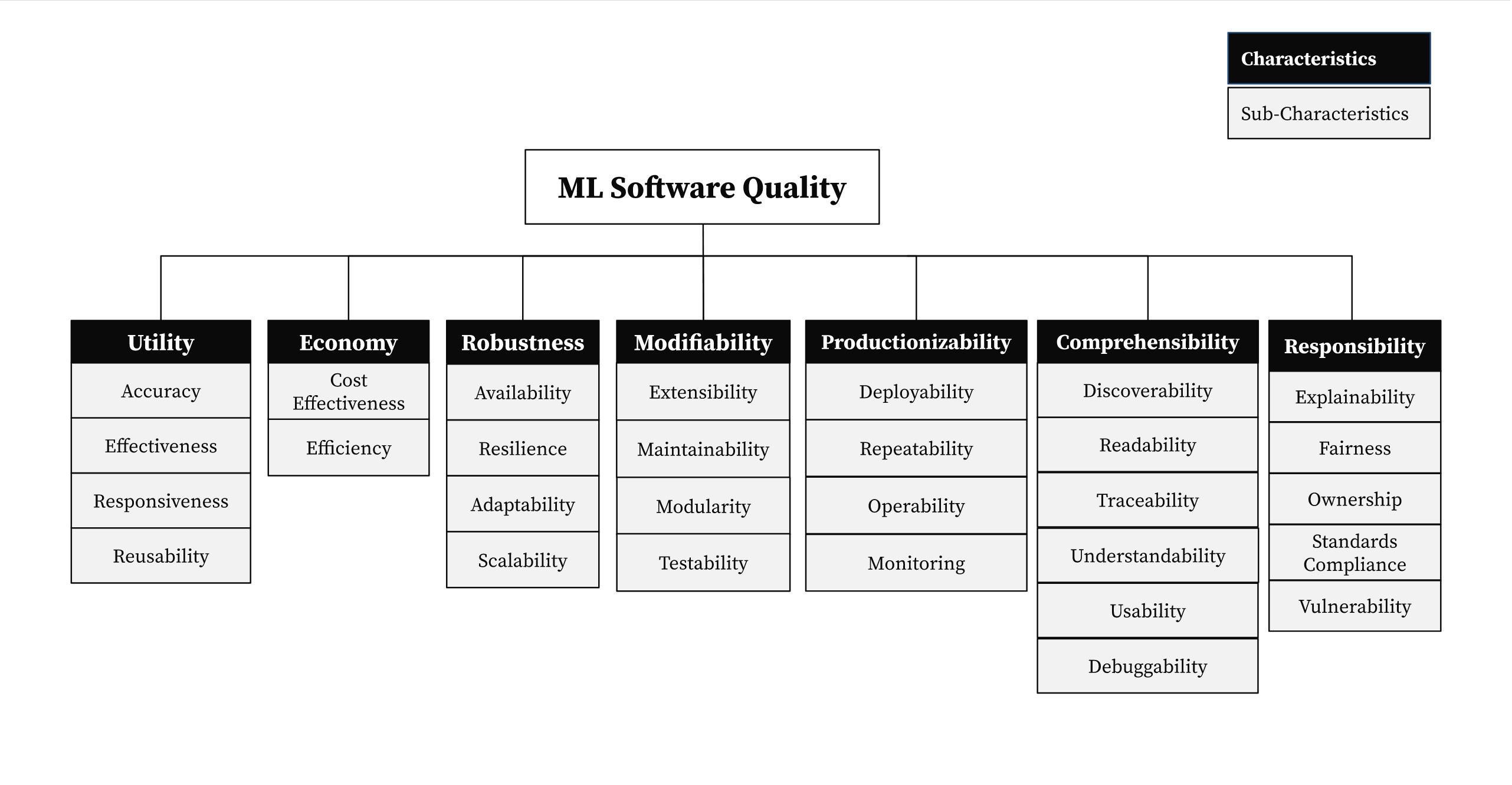}
  \caption{A software quality model for machine learning systems.}
  \label{fig:qm}
\end{figure*}
The definitions of all \textit{sub-characteristics} can be found in Appendix \ref{appendix:subchar_definitions}. Notice that there are no data quality attributes in the quality model, as these are defined in well established software quality models tailored for data \cite{ISO25012}. This existing data quality model can be used in addition to our software quality model, to analyze the quality of data which are used as input to an ML system.

\subsection{The development process}
We started by creating a list of the quality sub-characteristics to be included in our model. To achieve this, we went through the list of all the known system quality attributes in \cite{List_of_system_quality_attributes} and all software quality models in \cite{Miguel} from which we shortlisted and adapted the ones we judged applicable to machine learning systems. The shortlisting was done based on the relevance of each quality attribute to any stages of the ML development lifecycle defined in \cite{mllifecycle} and taking into account the various types of ML use cases e-commerce platforms like Booking.com has. Next, we added attributes related to machine learning that were not part of the initial list, such as \textit{fairness} and \textit{explainability} (as defined in Appendix \ref{appendix:subchar_definitions}). With the final list of attributes, we created clusters of factors (\textit{characteristics}) comprising related sub-factors (\textit{sub-characteristics}), following the standard nomenclature for quality models \cite{Miguel}.

We validated the completeness of our quality model using published sets of machine learning practices \cite{HiddenTechDebt}, \cite{Rubric}, \cite{Amershi}, \cite{Serban}, \cite{se_ml_website}. Concretely, we checked if we can relate these practices to at least one of the quality \textit{sub-characteristics} in our quality model. We iterated on this procedure a few times before we concluded on an first version, which was further refined using feedback from 10 internal senior ML engineers and scientists working in the industry and building ML systems for a minimum of 5 years. Given the speed with which the field is evolving, it is important to remark that the software quality model for machine learning is a live artifact constantly reviewed and updated in order to keep its relevance to the current machine learning needs. Another development process for a quality model for machine learning has been presented in \cite{Siebert2021}, in which the authors explain the implementation process of quality models for particular machine learning related use cases. Our development process aimed at creating a general-purpose quality model which is relevant for a wide range of machine learning applications. Different applications and organizations will put different emphasis onto different \textit{sub-characteristics} (for example external facing systems should be invulnerable even at the cost of accuracy) something that can be achieved by using importance weights per quality \textit{sub-characteristic}. Having a common quality model for all the machine learning systems allows its usage as a common language for quality related initiatives and for identification of gaps on quality attributes both at the system and organizational level.
\section{A framework to prioritize software practices}
\label{section:framework}
Choosing practices in order to improve ML quality is a challenging task mainly due to their large number, varying implementation costs, and overlapping effects. To tackle this, we propose a framework to analyze and prioritize software practices. Given a Software Quality Model represented by a set of sub-characteristics $C$, and a set of software best practices $P$ we want to choose a subset of practices maximizing the coverage of a given set of sub-characteristics, under a constraint of implementing at most $B$ practices \footnote{To simplify, we focus on the number of practices as cost function, but it is straightforward to extend to a general knapsack constraint \cite{Sviridenko} such as number of hours needed to adopt a practice.}. Having an influence $u(p, c)$ for a practice $p$ on a sub-characteristic $c$ we can define \textit{coverage} as a minimum threshold $k$ of influence. Formally we have:

\begin{enumerate}
    \item A Software Quality Model, represented by its set of \textit{sub-characteristics} $C$
    \item A set of software practices $P$
    \item For each practice $p \in P$ and each quality \textit{sub-characteristic} $c \in C$, the influence defined by a function $u: P \times C \rightarrow \mathbb{R}^{+}$
    \item A \textit{sub-characteristic} importance vector $w \in [0,1]^{|C|}$ representing the relevance of each \textit{sub-characteristic} $c \in C$
    \item An effort budget in the form of number of practices to be adopted $B \in \mathbb{N}$
    \item An integer $k$ representing the minimum influence necessary to consider any \textit{sub-characteristic} \textit{covered}
\end{enumerate}

We define the \textit{coverage function} as a set function that given a set of \textit{sub-characteristics} $C$ with importance weights $w$ and a coverage threshold $k$ maps a set of practices $X \in 2^P$ to a real number, formally:

\begin{equation}
\label{eq:coverage_function}
f(X; C, w, k) = \sum_{c \in C} w_c \min(k, \sum_{p \in X}  u(p,c))
\end{equation}

The objective is to choose a subset of practices that maximizes the coverage of the quality model weighted by its importance under the budget constraint: 
\begin{align}
     \argmaxA_{X \in 2^P} f(X; C, W, k) 
     \text{ subject to } |X| \leq B.
\end{align}

\subsection{Eliciting the relationship between best-practices and quality sub-characteristics}
\label{constructing}
In order to apply the framework in practice, we first needed a set of practices $P$. To achieve this, we conducted a survey with our internal ML practitioners at Booking.com where we asked them which 3 best practices for ML systems, from the ones they apply in their day to day work, they find the most useful. In total we received 25 responses from ML engineers and scientists with a minimum of 3 years of industrial experience building production ML systems. Based on the responses we created a list of 41 practices, which can be found in Appendix \ref{appendix:internal_practices}. Then, we obtained the values of the function $u(p, c)$ to be used as inputs in the framework by going through the following procedure.

We conducted a workshop with 13 internal ML practitioners (ML engineers and scientists with a minimum of 3 years of industrial experience building ML systems) who were given a lecture on the proposed Software Quality Model and had interactive exercises to ensure a deep understanding of all the quality sub-characteristics and their nuances. In the end of the workshop, the practitioners were given a quiz to assess their understanding. After the quiz, the practitioners were asked to score the set of 41 practices against each quality sub-characteristic ($C$) on a 0-4 scale indicating their influence: irrelevant (0), weakly contributes (1), contributes (2), strongly contributes (3) and addresses (4) \footnote{The scoring instructions can be found in Appendix \ref{appendix:scoring_instructions}.}. Finally by taking the median of the scores of all the practitioners we obtain the influence of each practice $p$ on each quality sub-characteristic $c$,  $u(p,c)$. To make this more concrete, we provide some examples of scores $u(p,c)$ for several pairs of quality sub-characteristic and practices in Table \ref{table:example_scores}. Influence scores for each sub-characteristic can be found in Appendix \ref{appendix:u_p_c_scores}.

Given the influence per practice and sub-characteristic $u(p,c)$ and a coverage threshold $k$, we can determine when a sub-characteristic is considered \textit{covered}. For example, given that we want to cover \textit{Understandability}, if $k=10$ then the practices \textit{documentation}, \textit{peer code review} and \textit{error analysis} with influence scores $u(p,c)$ of 4,3 and 3 respectively, do cover it. However the practices \textit{logging of metadata and artifacts}, \textit{data versioning} and \textit{alerting}, with influence scores of 2,1 and 0 respectively, do not cover \textit{Understandability}.

\begin{table}[!t]
\renewcommand{\arraystretch}{1.3}
\caption{Examples of Influence Scores $u(p,c)$}
\label{table:example_scores}
\begin{tabular}{ccc}
\textbf{Sub-characteristic} & \textbf{Practice}         & \textbf{Score} \\ \hline
Deployability   & Data Versioning & 0  \\
Repeatability & Documentation & 2 \\
Debuggability & Logging of Metadata And Artifacts & 3 \\
Traceability & Data Versioning & 3 \\
Understandability & Documentation & 4           
\end{tabular}
\end{table}

\subsection{Scaling of Influence Scores}
\label{scaling}

Based on ML practitioners' evaluation, four practices scored with an influence of \textit{weakly contributes} = 1  should not be treated equally as a practice scored with  \textit{addresses} = 4, hence to penalize weak contributions we re-scale the scores. To achieve this we chose a piecewise linear function where we define the \textit{addresses} influence score = 4*\textit{strongly contributes}, \textit{strongly contributes} = 3 * \textit{contributes}, \textit{contributes} = 2 * \textit{weakly contributes}. For continuous values, after averaging multiple ML practitioners scores, we apply a piecewise linear function between these values which we depict in Figure \ref{fig:scaling_function}.

We defined \textit{coverage} in Equation \ref{eq:coverage_function} as the minimum threshold of influence $k$. We chose one \textit{addresses} influence to \textit{cover} a \textit{sub-characteristic}, and after applying our re-scaling function we get $k = 24$. In general, the parameter $k$ defines the coverage threshold, and the re-scaling allows to parameterize the relationship of the influence scores while keeping the scoring of the \textit{sub-characteristic} and practice pairs on a small linear scale of $[0; 4] \in \mathbb{Z}^{0+}$.

The choice of $k$ and of the re-scaling function depend on the application where the ML System is deployed and on the risk of wrongly treating a \textit{sub-characteristic} as covered. 

\begin{figure}[H]
 \centering
 \includegraphics[width=0.4\textwidth]{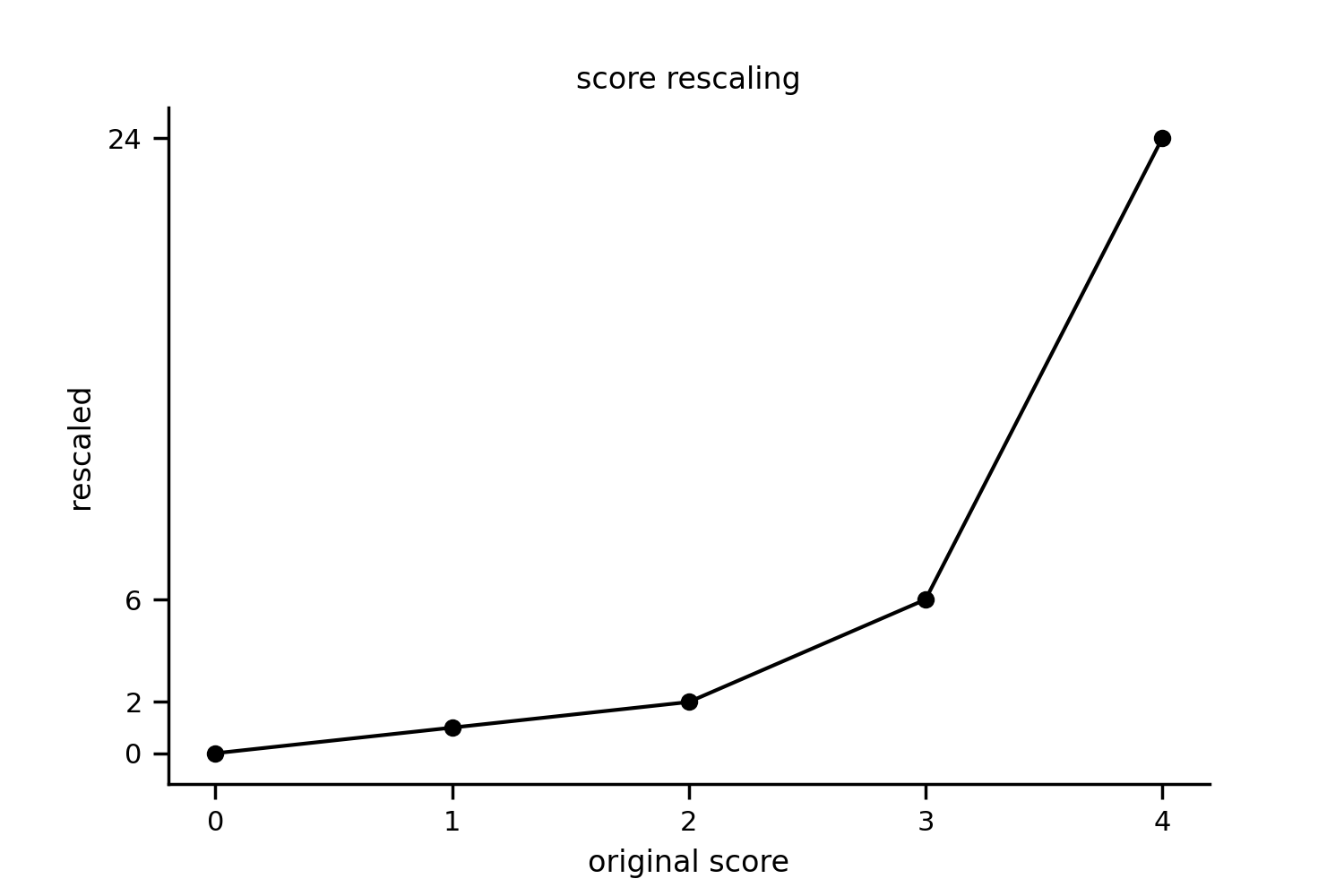}
 \caption{Scaling function for influence scores}
 \label{fig:scaling_function}
\end{figure}

\subsection{Inter-annotator Agreement}
\label{section:agreement}

Assessing the influence of a practice in a quality sub-characteristic is a subjective task and therefore subject to annotator disagreement. We used two tests for agreement - whether two scores are identical (referred as \textit{plain} agreement) and whether two scores differ by more than one level (referred as \textit{practical} agreement). The \textit{practical} test is more aligned with the complexity of the task and the variance coming from the practitioners experience and knowledge.
We found an average agreement rate (between a pair of annotators) of 73.56\% (plain) and 86.38\% (practical). We used Cohen's Kappa to check the agreement rate while neutralizing the probability of agreement happening by chance, and reached 0.4 (plain) and 0.69 (practical). These scores represent an agreement rate which is between fair (\textit{plain}) and substantial (\textit{practical}) according to \cite{Viera}.  

The observed consistency suggests that we can have new best practices sets (or new quality sub-characteristics), scored by substantially fewer practitioners, which we consider an important insight when it comes to adopting new practices in an industrial setting. For example, considering the case of only two annotators, we estimate the sampling distribution for both the agreement-rate and Kappa statistic by computing the metric for every possible pair of annotators among the 13. For the agreement rates, the standard deviation is 1.38\% (\textit{plain}) and 1.68\% (\textit{practical}), and for the Kappa statistic the standard deviation is 0.043 (\textit{plain}) and 0.05 (\textit{practical}). Both figures are low enough which enables us to substitute a large group of annotators with only a pair and still get reliable scores.

\subsection{Algorithms}
\label{section:algorithms}

The maximization problem we want to solve is similar to the Generalized Maximum Coverage (GMC) problem
\cite{GMC}, with a clear difference: in GMC if a set $X$ covers an element $a$, then at least one subset $Y \subset X$ covers $a$. In our case, if a set of practices $Q \subseteq P$ covers a sub-characteristic $c \in C$, it might be the case that no subset of $Q$ covers $c$. Consider two practices $p_1, p_2$ and sub-characteristic $c$ with $u(p_1, c) = u(p_2, c)=k/2$. In this case the set $Q=\{p1, p2\}$ covers $c$ since $f(Q;\{c\}, 1, k))=k$ but no subset of $Q$ does since $f(\{p1\};\{c\}, 1, k))=f(\{p2\};\{c\}, 1, k))=k/2$ and $f(\emptyset;\{c\}, 1, k))=0$. Because of this, a specific analysis is required.

The budget expressed as the maximum number of practices to be applied leads to a combinatorial explosion of the search space. To illustrate, the set of 41 practices we collected and a budget of 3 practices yields a search space of size $41 \choose 3$ = 10660, whereas a budget of 10 practices yields a search space of $1.12\mathrm{e}{+9}$ options to explore. To tackle this computational problem we propose a greedy solution based on the observation that $f$ is positive monotone submodular (proof in Appendix \ref{appendix:proof}). Maximizing a monotone submodular function is known to be NP-Hard \cite{NP_Article}, \cite{NP_Book}, however a simple greedy approach yields a $(1 - \frac{1}{e})$-approximation \cite{Nemhauser} even for one general knapsack constrain \cite{Sviridenko}, and it is the best polynomial time solution, unless $P=NP$ \cite{Nemhauser_best}, \cite{feige}. We propose two solutions: brute force and greedy, in Algorithm \ref{alg:brute} and \ref{alg:greedy} respectively. In practice we found that the greedy approach rarely yields sub-optimal results for this case.

\begin{minipage}{0.46\textwidth}

\begin{algorithm}[H]
\caption{Brute force search}\label{alg:brute}
  \small
\begin{algorithmic}[1]
\Require $B$: budget of practices to be used
\Require $W$: sub-characteristic importance vector
\Require $C$: set of sub-characteristics
\Require $P$: set of practices
\Require $k$: coverage threshold
\State $score \gets 0$
\State $selected \gets \emptyset$
\For{each possible subset $P_i \subseteq
 P$ of size $B$}
    \State $curr \gets f(P_i; C, W, k)$
    \If{$curr > score$}
        \State $curr \gets score$
        \State $selected \gets P_i$
    \EndIf
\EndFor
\State \Return $selected$
\end{algorithmic}
\end{algorithm}

\end{minipage}

\hfill
\algnewcommand{\LineComment}[1]{\State \(\triangleright\) #1}

\begin{minipage}{0.46\textwidth}

\begin{algorithm}[H]
\caption{Greedy algorithm}\label{alg:greedy}
    \small  
\begin{algorithmic}[1]
\Require $B$: budget of practices to be used
\Require $W$: sub-characteristic importance vector
\Require $C$: set of sub-characteristics
\Require $P$: set of practices
\Require $k$: coverage threshold
\State $S \gets \emptyset$ \Comment{set with selected practices}
\State $all\_covered \gets False$ \Comment{are all sub-chars covered?}
\State $cov \gets [0$ for each $c \in C]$ \Comment{track coverage for each sub-char}
\While{$|S| < B$ and $\neg all\_covered$ and $|P \setminus S|>0$} 
        \State $p_{s} \gets \argmaxA_{p \in P \setminus S } {f(S \cup \{p\}; C, W, k)}$ \Comment{find best practice }    
        \State $S \gets S \cup \{p_{s}\}$ \Comment{greedily add best practice to selection}
    \For {$c \in  C $} \Comment{recompute coverage with new selection}
        \State $cov[c] = f(S; \{c\}, W, k)$ 
    \EndFor
    \State $all\_covered \gets \sum{cov} == |C| k$ \Comment{are all sub-chars saturated?}
\EndWhile
\State \Return $S$
\end{algorithmic}
\end{algorithm}

\end{minipage}

\section{Applying the framework}
\label{section:application}
In this section we illustrate the usage of our framework by analyzing our own best-practices set and three well-known ML best-practices sets \cite{Rubric}, \cite{Amershi}, \cite{Serban} and \cite{se_ml_website} (we combine the last two as they intersect) including 28, 7, and 45 best practices respectively. In each case we compute the coverage function, optimal practices sets for different budgets, and highlight gaps as well as general trends. We also provide a \textit{global} analysis combining all sets of best practices. 

\subsection{Analyzing sets of best practices}
\label{sec:analyzing_best_practices}

\subsubsection{Internal Set}
    Using the influence vectors of the internal set of 41 practices applied at Booking.com, we can visualize the total contribution of the set to all the quality sub-characteristics and assess its completeness. We plot the contributions of the internal set in Figure \ref{fig:practices_internal}, where we mark the threshold $k=24$ contribution points indicating coverage of a quality sub-characteristic. We observe that 22 out of 29 sub-characteristics are being covered indicating a coverage rate of 75\%. The sub-characteristics with the largest influences are mostly associated with traditional software systems, such as \textit{effectiveness} and \textit{monitoring}, while the ones with the least influences are more specific to ML systems, such as \textit{explainability} and \textit{discoverability}. This is due to the fact that historically, engineering best practices are more closely related to traditional software systems and only in the recent years ML specific best practices started becoming popular. Based on this analysis we were able to identify the areas for which practices are lacking and work towards their coverage, by creating new ones. Concretely, to address the gaps in \textit{Vulnerability}, \textit{Responsiveness} and \textit{Discoverability} we created the following practices: 
 "Request an ML system security inspection", "Latency and Throughput are measured and requirements are defined", "Register the ML system in an accessible registry", which increase the coverage for each of the sub-characteristics respectively (see Appendix \ref{appendix:internal_practices} for their descriptions).

\begin{figure}[!t]
\centering
\includegraphics[width=0.48\textwidth]{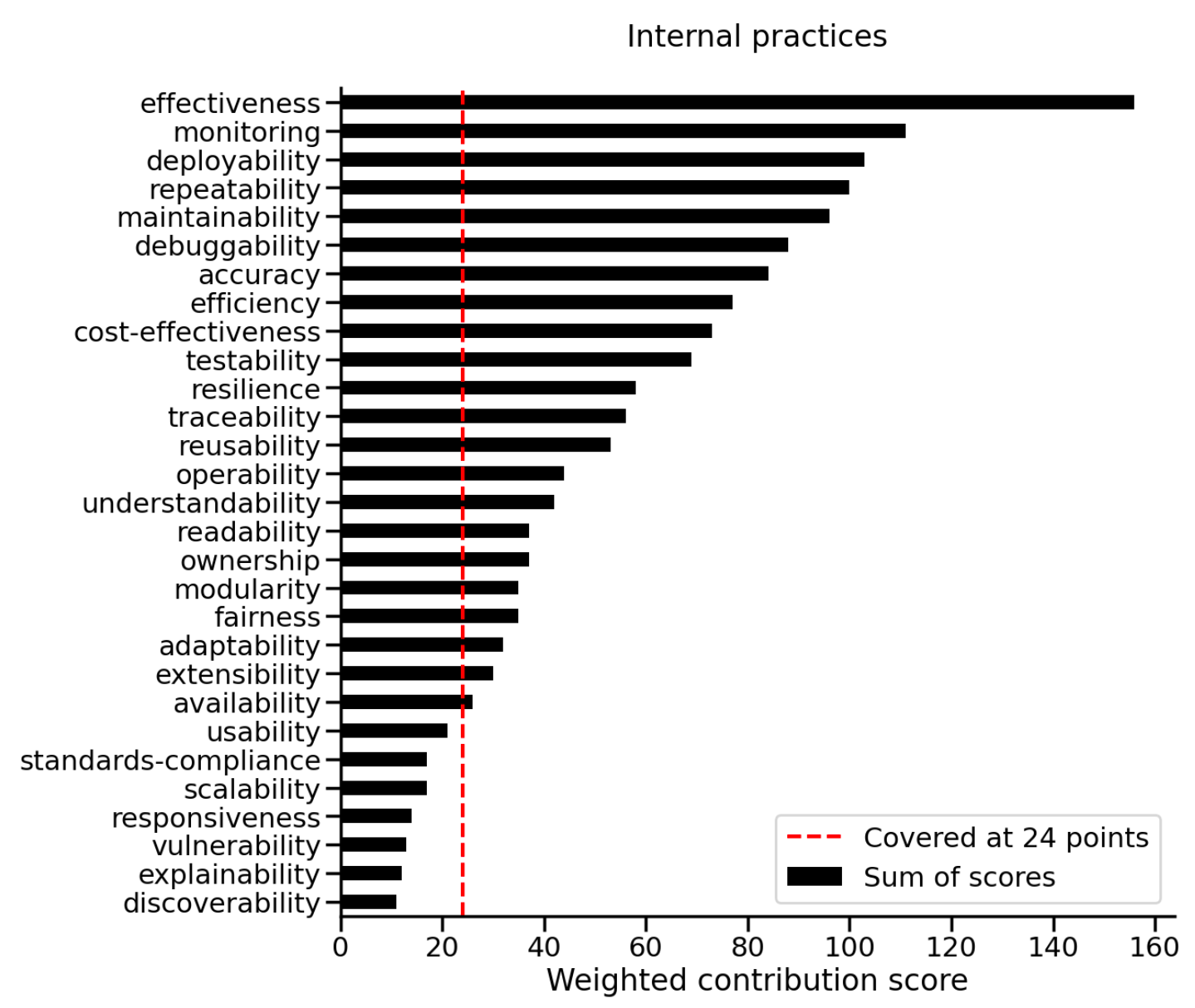}
\caption{Coverage of the quality sub-characteristics by applying all the 41 practices from the internal set.}
\label{fig:practices_internal}
\end{figure}

To gain further insight, we use the Greedy algorithm to find the top 3 influential practices on all quality sub-characteristics, considering them all equally important. The algorithm outputs a set of the following top 3 practices: "Write documentation about the ML system", "Write modular and reusable code", and "Automate the ML lifecycle". This result has been used to guide the ML practitioners at Booking.com on the prioritization of practice adoption in their daily work, by highlighting the value of these practices on the overall ML quality. The actual prioritization of their adoption depends on the team, since different teams and departments use different priorities for the quality sub-characteristics.   

\subsubsection{External Sets}
We analyze three ML best practices sets of 80 practices in total. Since it is impractical to have the same 13 ML practitioners scoring the 80 practices, we limit the number of annotators to 2, based on the high agreement rate for a pair of annotators observed in Section \ref{section:agreement}. After the scoring, we compute the plain agreement rate for the 2 annotators to be 63.5\% and the practical agreement rate 94.5\%. With these vectors, we can visualize the total contribution of the whole set of practices to each of the quality sub-characteristics and based on that assess which of them are being covered.  In Figure \ref{fig:practices_serban} we see that applying all the practices presented in \cite{Serban} 25 sub-characteristics are covered. In this set of practices the strongest emphasis is on sub-characteristics related to \textit{cost-effectiveness}, \textit{responsibility} and \textit{modifiability}. On the other hand, \textit{sub-characteristics} such as \textit{scalability}, \textit{discoverability}, \textit{operability} and \textit{responsiveness}, remain uncovered even when applying all the 45 practices from this set. Figure \ref{fig:practices_breck} illustrates the contributions by applying all the 28 practices mentioned in \cite{Rubric} and we observe that this set covers 17 \textit{sub-characteristics}: we observe the top contributions to be on non-ML specific quality \textit{sub-characteristics}, although ML specific ones such as \textit{accuracy} and \textit{fairness} are also covered. The least covered are related to collaboration such as \textit{ownership}, \textit{discoverability} and \textit{readability}. Lastly, the contributions of \cite{Amershi} to the software quality are depicted in Figure \ref{fig:practices_amershi}. This set of 7 practices manages to cover 9 quality \textit{sub-characteristics} with a focus on those related to \textit{economy} and \textit{modifiability}. The least contributions are achieved on aspects related to the \textit{comprehensibility} of ML systems. In general we find that all practice sets focus on different quality attributes and have gaps on different areas of our SQM. This indicates that the sets complement each other, which motivates our next analysis.

\begin{figure*}[!t]
  \centering
  \begin{subfigure}{0.48\textwidth}
         \centering
         \includegraphics[width=\textwidth]{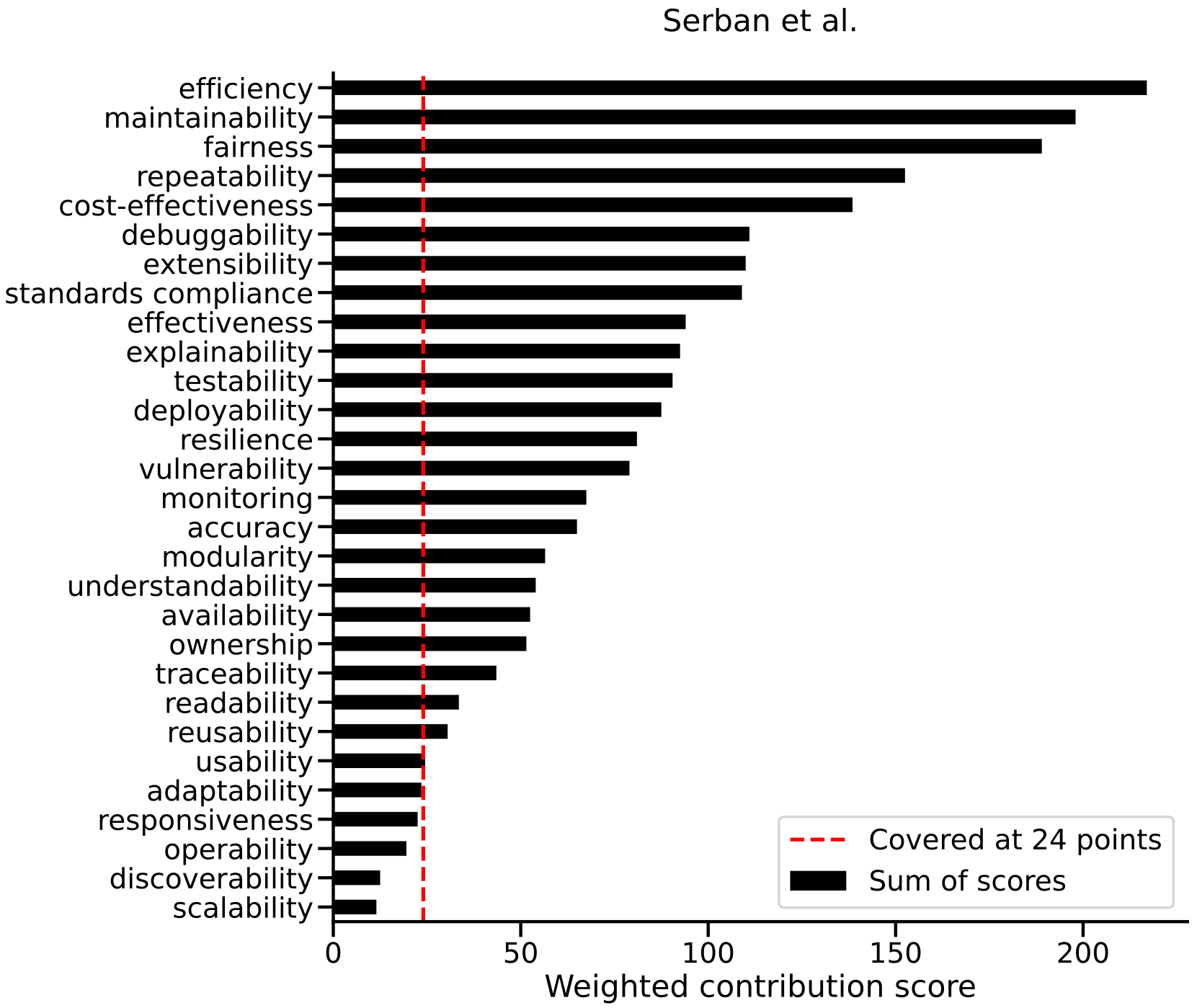}
         \caption{}
         \label{fig:practices_serban}
  \end{subfigure}
     \hfill
     \begin{subfigure}{0.48\textwidth}
         \centering
        \includegraphics[width=\textwidth]{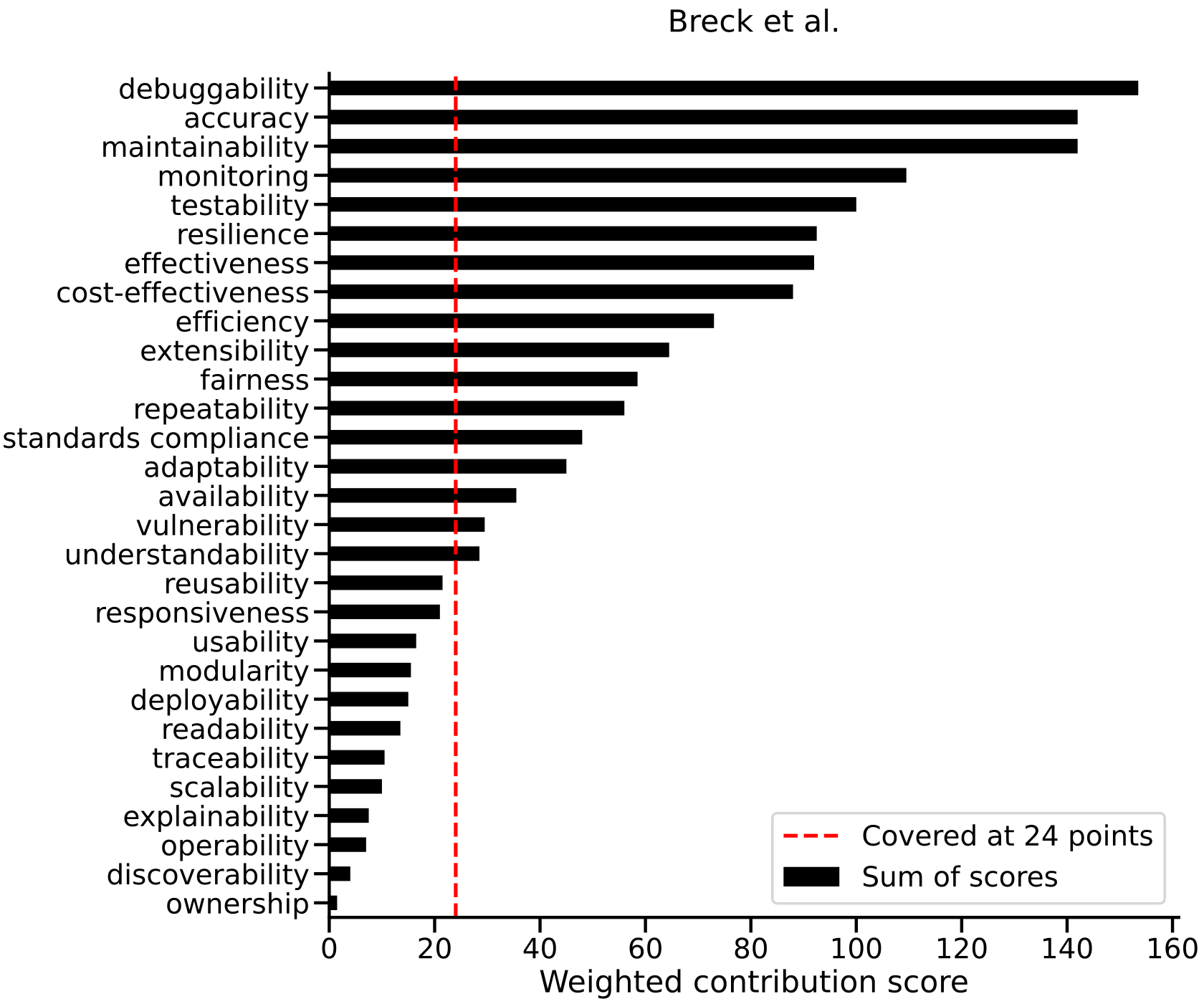}
         \caption{}
         \label{fig:practices_breck}
    \end{subfigure}
    \hfill
    \begin{subfigure}{0.48\textwidth}
         \centering
  \includegraphics[width=\textwidth]{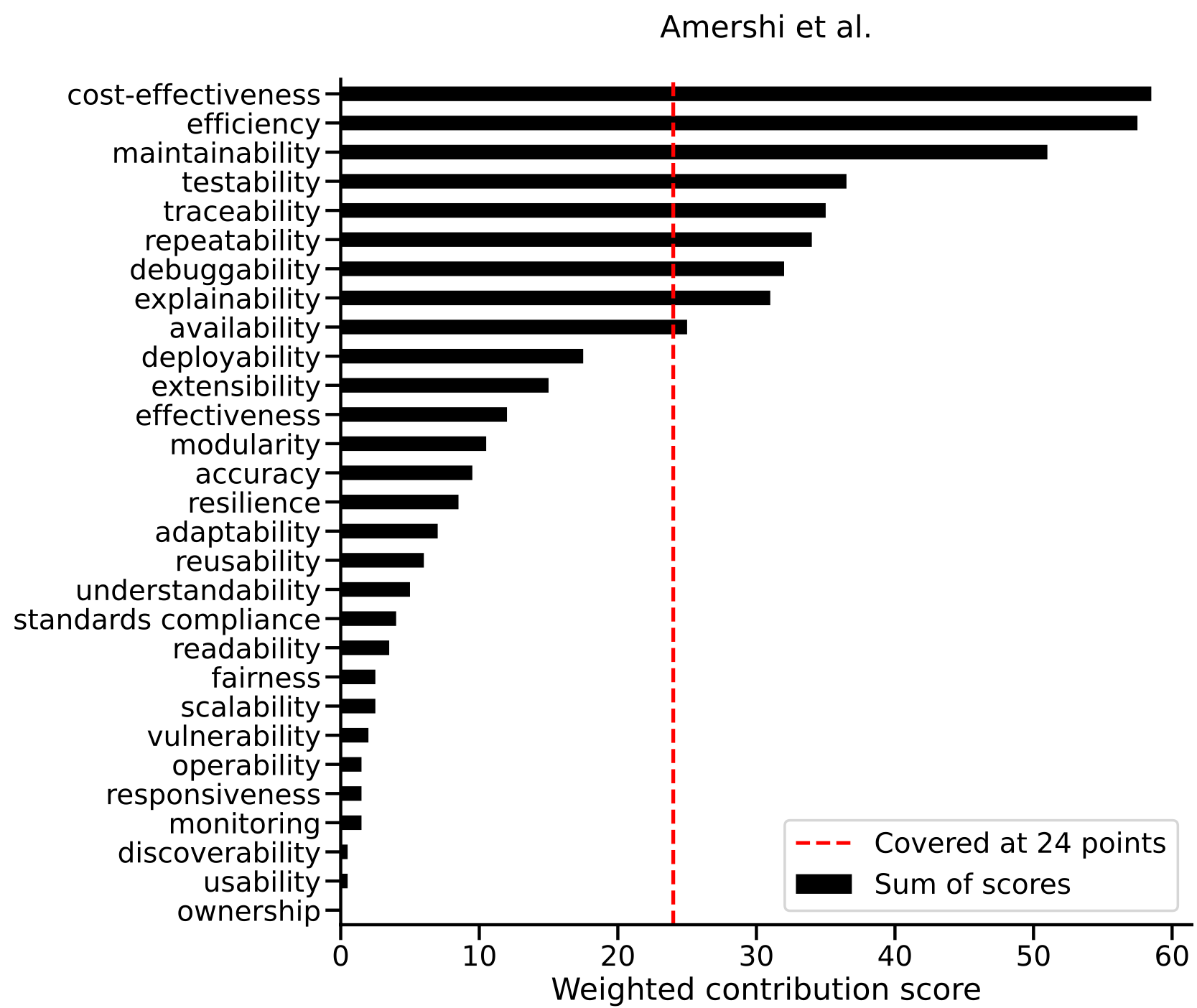}
        \caption{}
         \label{fig:practices_amershi}
    \end{subfigure}
    \hfill
    \begin{subfigure}{0.48\textwidth}
         \centering
    \includegraphics[width=\textwidth]{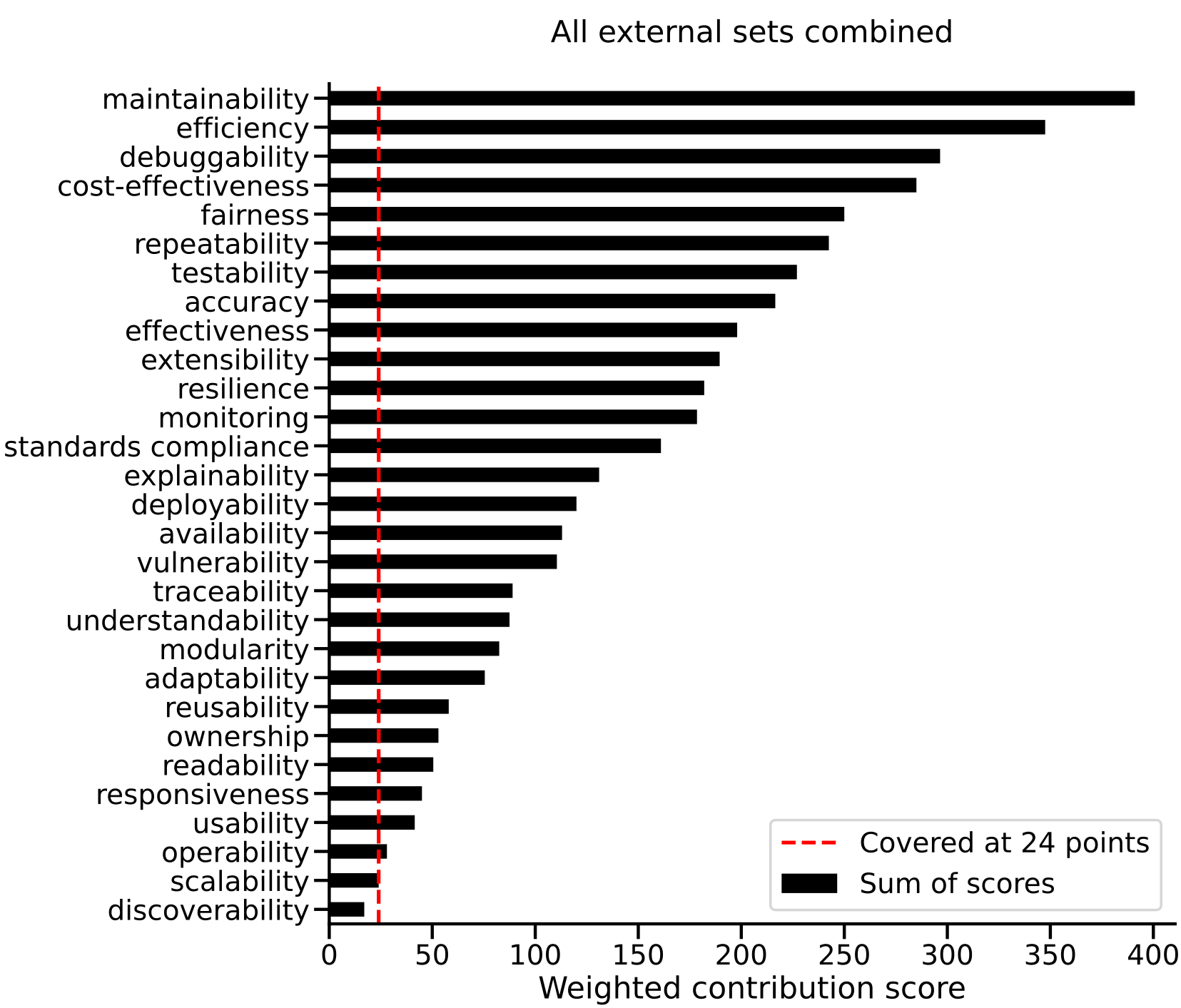}
         \caption{}
         \label{fig:practices_all}
    \end{subfigure}
  \caption{Coverage of the quality sub-characteristics by applying all the best practices in \cite{Serban} (\ref{fig:practices_serban}), \cite{Rubric} (\ref{fig:practices_breck}), \cite{Amershi} (\ref{fig:practices_amershi}) and the combined set (\ref{fig:practices_all}).}
  \label{fig:practices_analysis}
\end{figure*}

In Figure \ref{fig:practices_all} we look into the quality coverage in the scenario where we apply all the practices combined. After removing overlapping practices (see Appendix \ref{appendix:overlapping}), this set includes 76 practices. We observe that when we apply the full set of 76 practices, 28 sub-characteristics are covered which verifies that the practices complement each other. An example that shows this is \textit{scalability}, which is not covered by any set in isolation, but only when the practices are combined. We also see that even when applying all the 76 practices, \textit{discoverability} remains uncovered. This shows that there is lack of practices addressing this quality sub-characteristic, something that was also observed in the analysis of the internal practice set. Moreover, the low scores for sub-characteristics like \textit{scalability}, \textit{operability}, \textit{usability} and \textit{responsiveness} indicate that they receive less attention compared to the rest. On the other hand, it is encouraging to see large scores for sub-characteristics related to trustworthiness such as \textit{fairness} and \textit{explainability}.  

\subsection{Score and coverage threshold sensitivity}
\label{sec:score_sensitivity}

To further assess the sensitivity of the results to the scores assigned by the ML practitioners, we perturb the scores by adding a random integer in the range $[-1; 1]$ and $[-2; 2]$. We then take the original scores and perturbed ones, and compute the scores of each \textit{sub-characteristic} as if all practices were applied and rank them by the sum of scores. Then we measure the Pearson correlation coefficient of the original ranking and the ranking after the scores were perturbed. After 1000 perturbation iterations we obtain a mean correlation coefficient of 0.94 with a variance of 0.0002 for perturbing by $[-1; 1]$, and a mean of 0.91 with a variance of 0.0006 for perturbing by $[-2; 2]$ respectively. A random integer in the range $[-3; 3]$ yields a mean of 0.86 and a variance of 0.0016. This shows that our results are robust to scoring variance. Regarding the coverage threshold $k$ we remark that 24 points is rather low since one single practice with \textit{addresses} score would cover the sub-characteristic, at the same time, in Figures \ref{fig:practices_internal} and \ref{fig:practices_analysis} we can see that small changes in $k$ do not lead to big changes in which quality sub-characteristics are covered, more importantly, the general observations hold even for moderate changes in $k$.

\subsection{How many practices are enough?}
\label{sec:how_many_practices}

To evaluate how many practices are enough to maximize quality, we analyze the internal and open source sets combined (after removing overlapping practices the combined set has 101 practices, see Appendix \ref{appendix:overlapping} for details). Using our prioritization framework we find the minimum number of practices which cover the same number of quality \textit{sub-characteristics} as the full set of those 101 practices. To achieve that, we find the top $N$ practices from the combined set of practices using our greedy algorithm (brute force takes too long), for $N \in [1,101]$ and we evaluate what percentage of the quality \textit{sub-characteristics} is being covered with each set of practices. Figure \ref{fig:percent_covered} illustrates the coverage percentage for all the values of $N$. We see that applying 5 practices covers almost $40\%$, 10 cover $70\%$, and to reach $96\%$, 24 are needed. The coverage does not increase further with the practices. This result shows that using a relatively small number of practices can achieve similar results in terms of quality coverage to the full set of 101 practices. This means that when applying the right set of practices, a significant reduction in the effort of adoption can be achieved, which is especially relevant in an industrial setting.

\begin{figure}[h]
  \centering
  \includegraphics[width=0.4\textwidth]{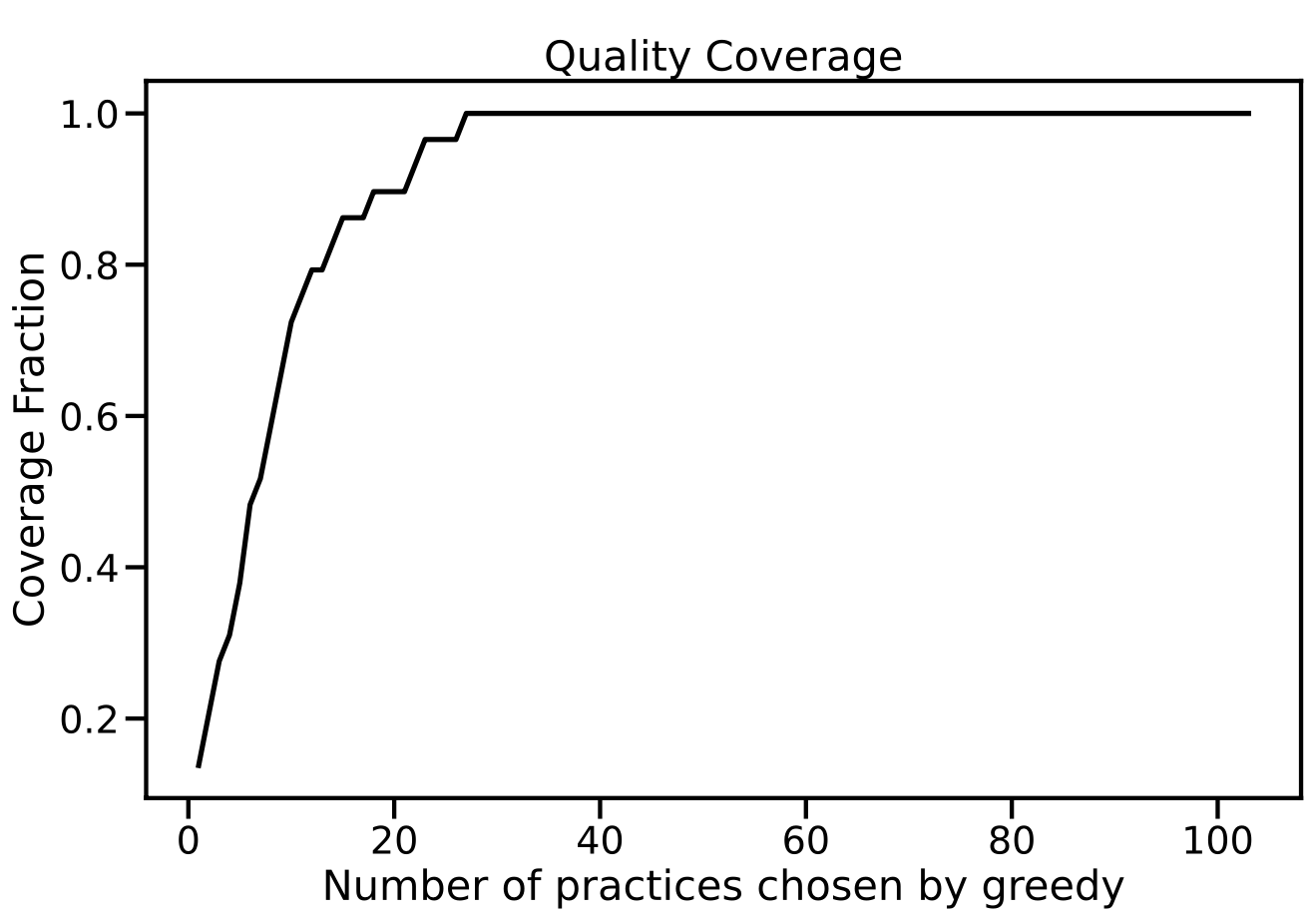}
  \caption{Coverage of the quality sub-characteristics for an increasing budget of practices. The practices to be applied are chosen using the greedy algorithm (\ref{alg:greedy}), on the combined set of open source and internal practices for ML systems.}
  \label{fig:percent_covered}
\end{figure}

\subsection{Which are the best practices?}
\label{section:which_are_the_best}

To gain further insights as to which are the 24 practices which maximize coverage, we provide the optimal set in Table \ref{table:practices_max_quality}, along with the source of each practice (some practices have been renamed for better clarity, see Appendix \ref{appendix:explanation_of_best_practices} for details). It is important to note that here we assume equal importance for each quality sub-characteristic, something that needs to be taken into account from ML practitioners wanting to use this set as guidance. In case a different importance weighting is desired, one needs to re-create this set after applying importance weights to each sub-characteristic. Prioritization within the final set, can be achieved by taking into account the specific needs of an organization (for example if safety is top priority, practices focusing on robustness should be prioritized) or the cost of adoption per practice. 

\begin{table}[]
\caption{The practices which maximize quality.}
\label{table:practices_max_quality}
\begin{tabular}{l}
\textbf{Practice Name \& Source}                                                    \\
Versioning for Data, Model, Configurations and Scripts \cite{Serban}   \\
Continuously Monitor the Behaviour of Deployed Models               \cite{Serban}   \\
Unifying and automating ML workflow                                 \cite{Amershi}  \\
Remove redundant features.                                      \cite{Rubric}   \\
Continuously Measure Model Quality and Performance                  \cite{Serban}   \\
All input feature code is tested.                                   \cite{Rubric}    \\
Automate Model Deployment                                           \cite{Serban}  \\
Use of Containarized Environment                                [Appx. \ref{appendix:internal_practices}] \\
Unified Environment for all Lifecycle Steps                          [Appx.  \ref{appendix:internal_practices}] \\
Enable Shadow Deployment                                            \cite{Serban}   \\
The ML system outperforms a simple baseline.                                      \cite{Rubric}    \\
Have Your Application Audited                                       \cite{Serban}   \\
Monitor model staleness                                           \cite{Rubric}     \\
Use A Collaborative Development Platform                            \cite{Serban}   \\
Explain Results and Decisions to Users                              \cite{Serban}   \\
The ML system has a clear owner.                                     [Appx.  \ref{appendix:internal_practices}] \\
Assign an Owner to Each Feature and Document its Rationale          \cite{Serban}   \\
Computing performance has not regressed.                            \cite{Rubric}   \\
Communicate, Align, and Collaborate With Others                     \cite{Serban}   \\
Perform Risk Assessments                                            \cite{Serban}  \\
Peer Review Training Scripts                                        \cite{Serban}   \\
Establish Responsible AI Values                                     \cite{Serban}   \\
Write documentation about the ML system.                  [Appx.  \ref{appendix:internal_practices}] \\
Write Modular and Reusable Code                                      [Appx.  \ref{appendix:internal_practices}]
\end{tabular}
\end{table}

\subsection{Further applications of the framework}
\label{section:usages}
The proposed SQM is currently being used to construct a quality assessment framework for ML systems. Concretely, the framework assesses the coverage of each quality sub-characteristic on an ML system level, to pinpoint improvement areas. Implementing an ML quality assessment framework without an SQM for ML systems, would lead to an incomplete picture of ML quality. Moreover, the prioritization framework is being used alongside the quality assessment framework: After the quality of an ML system is assessed, by assigning a quality score per quality sub-characteristic, the sub-characteristics with low scores are provided as input in the prioritization framework in order to recommend the best 3 practices to apply in order to cover them. This has been very helpful for ML practitioners as it allows them to prioritize the improvements to be made efficiently, by focusing on practices that have the largest influence in the quality attributes that are considered the most important for the use case at hand. 

Additionally, the SQM has created a common language for ML practitioners to discuss ML quality topics and quality related initiatives are easier to be justified. For example, it is more straightforward to argue about the value of an initiative targeting to increase the adoption of unit-testing for ML systems, since the benefit of it, e.g. improvement in \textit{modifiability} of the system, is clear.

An advantage of our framework is that it is flexible enough to be adapted to other organizations. For completeness, we describe how this can happen. The organization needs to determine which quality \textit{sub-characteristics} are the most crucial, by specifying the importance weights $W$ for each \textit{sub-characteristic}. The provided software practices can be used as is or new ones can be added and scored by ML practitioners within the organization. Lastly, a coverage threshold $k$ should be chosen based on how strict an organization wants to be for solving a given quality \textit{sub-characteristic}. To deal with disagreements in the scores $u(p,c)$ or the coverage threshold $k$, the mean or median can be taken. Then, all an ML practitioner needs to do is to run the prioritization algorithm using as inputs the quality \textit{sub-characteristics} $C$ to be improved, the set of practices $P$ to be considered, the allowed budget $B$, the importance vectors $W$ and the coverage threshold $k$, and then adopt the optimal practices which are recommended by the framework. 

\section{Conclusions and Discussions}
\label{section:conclusion}

\textbf{Conclusion.}
In this work we presented a framework to analyse the relationship between software engineering best practices for ML Systems and their quality with the primary purpose of prioritizing their implementation in an industrial setting. We addressed the challenge of defining quality by introducing a novel Software Quality Model specifically tailored for ML Systems. The relationship between best practices and the various aspects of quality was elicited by means of expert opinion and represented by vectors over the sub-characteristics of the Software Quality Model. With these vectors we applied Set Optimization techniques to find the subset of best practices that maximize the coverage of the SQM. We applied our framework to analyse 1 large industrial set of best practices as implemented at Booking.com and 3 public sets. Our main findings are:
\begin{enumerate}
    \item Different best-practices sets focus on different aspects of quality, reflecting the priorities and biases of the authors.
    \item Combining the different best-practices sets, high coverage is achieved, remarkably, aspects that no single best-practices set covers on its own are covered by integrating different practices proposed by different authors.
    \item Even though there is a proliferation of best practices for ML Systems, when chosen carefully, only a few are needed to achieve high coverage of all quality aspects.
    \item Even though the influence of best-practices on quality aspects is a subjective concept we found surprisingly high consistency among experts.
\end{enumerate}
Our framework was useful to spot gaps in our practices leading to the creation of new ones to increase the coverage of specific quality aspects.

\textbf{Limitations.}
A limitation of this work is that in order to add a new quality \textit{sub-characteristic} or a new practice to the framework, one needs to score the influence vectors which is a time consuming procedure. On the other hand, the addition or removal of an existing practice or quality \textit{sub-characteristic} does not influence the existing scores. Another caveat regards the subjectivity of the influence vectors based on the individuals who conduct the scoring. However, our sensitivity analysis described in Section \ref{sec:score_sensitivity} indicates that our results are robust to scoring variance, which mitigates the subjectivity concerns. 

\textbf{Future Work.}
Future work will focus on a comparison of our framework with baseline prioritization approaches (such as prioritizing the most popular practices first or the ones requiring the least effort) and on assessing the coverage of \textit{sub-characteristics} in existing ML Systems. We will also keep evolving the assessment framework mentioned in Section \ref{section:usages} since this can provide visibility on quality gaps of ML systems, and along with the prioritization framework can provide guidance to ML practitioners on the optimal actions to take to improve them. Furthermore, exploring more realistic practice implementation cost functions can lead to a better cost and quality trade-off. Lastly, even though we aim at producing a complete software quality model, further validation is necessary especially by the external ML community.

\begin{acks}
We would like to thank the ML practitioners of Booking.com for being very helpful with providing input for important components of this work: the scoring of the influence vectors, the survey of best practices, the feedback on the SQM and the prioritization framework. Quality improvements can only be done successfully in collaboration with the practitioners, and without their help this work would not be possible.
\end{acks}

\bibliographystyle{ACM-Reference-Format}
\bibliography{prioritization_framework}


\begin{thebibliography}{35}


\ifx \showCODEN    \undefined \def \showCODEN     #1{\unskip}     \fi
\ifx \showDOI      \undefined \def \showDOI       #1{#1}\fi
\ifx \showISBNx    \undefined \def \showISBNx     #1{\unskip}     \fi
\ifx \showISBNxiii \undefined \def \showISBNxiii  #1{\unskip}     \fi
\ifx \showISSN     \undefined \def \showISSN      #1{\unskip}     \fi
\ifx \showLCCN     \undefined \def \showLCCN      #1{\unskip}     \fi
\ifx \shownote     \undefined \def \shownote      #1{#1}          \fi
\ifx \showarticletitle \undefined \def \showarticletitle #1{#1}   \fi
\ifx \showURL      \undefined \def \showURL       {\relax}        \fi
\providecommand\bibfield[2]{#2}
\providecommand\bibinfo[2]{#2}
\providecommand\natexlab[1]{#1}
\providecommand\showeprint[2][]{arXiv:#2}

\bibitem[Alvaro et~al\mbox{.}(2005)]%
        {Alvaro}
\bibfield{author}{\bibinfo{person}{Alexandre Alvaro},
  \bibinfo{person}{Eduardo~Santana de Almeida}, {and}
  \bibinfo{person}{Alexandre~Marcos Lins}.} \bibinfo{year}{2005}\natexlab{}.
\newblock \showarticletitle{Towards a Software Component Quality Model}.
\newblock


\bibitem[Amershi et~al\mbox{.}(2019)]%
        {Amershi}
\bibfield{author}{\bibinfo{person}{Saleema Amershi}, \bibinfo{person}{Andrew
  Begel}, \bibinfo{person}{Christian Bird}, \bibinfo{person}{Robert DeLine},
  \bibinfo{person}{Harald Gall}, \bibinfo{person}{Ece Kamar},
  \bibinfo{person}{Nachiappan Nagappan}, \bibinfo{person}{Besmira Nushi}, {and}
  \bibinfo{person}{Thomas Zimmermann}.} \bibinfo{year}{2019}\natexlab{}.
\newblock \showarticletitle{Software Engineering for Machine Learning: A Case
  Study}. In \bibinfo{booktitle}{\emph{2019 IEEE/ACM 41st International
  Conference on Software Engineering: Software Engineering in Practice
  (ICSE-SEIP)}}. \bibinfo{pages}{291--300}.
\newblock
\urldef\tempurl%
\url{https://doi.org/10.1109/ICSE-SEIP.2019.00042}
\showDOI{\tempurl}


\bibitem[Ashmore et~al\mbox{.}(2021)]%
        {mllifecycle}
\bibfield{author}{\bibinfo{person}{Rob Ashmore}, \bibinfo{person}{Radu
  Calinescu}, {and} \bibinfo{person}{Colin Paterson}.}
  \bibinfo{year}{2021}\natexlab{}.
\newblock \showarticletitle{Assuring the Machine Learning Lifecycle:
  Desiderata, Methods, and Challenges}.
\newblock \bibinfo{journal}{\emph{ACM Comput. Surv.}} \bibinfo{volume}{54},
  \bibinfo{number}{5}, Article \bibinfo{articleno}{111} (\bibinfo{date}{may}
  \bibinfo{year}{2021}), \bibinfo{numpages}{39}~pages.
\newblock
\showISSN{0360-0300}
\urldef\tempurl%
\url{https://doi.org/10.1145/3453444}
\showDOI{\tempurl}


\bibitem[Bertoa and Vallecillo(2002)]%
        {bertoa}
\bibfield{author}{\bibinfo{person}{Manuel~F. Bertoa} {and}
  \bibinfo{person}{Antonio Vallecillo}.} \bibinfo{year}{2002}\natexlab{}.
\newblock \showarticletitle{Quality Attributes for COTS Components}.
\newblock \bibinfo{journal}{\emph{Proceedings of the 6th ECOOP Workshop on
  Quantitative Approaches in Object-Oriented Software Engineering (QAOOSE
  2002)}} \bibinfo{volume}{1}, \bibinfo{number}{2} (\bibinfo{date}{November}
  \bibinfo{year}{2002}), \bibinfo{pages}{128--144}.
\newblock
\urldef\tempurl%
\url{http://www.lcc.uma.es/~av/publicaciones.html}
\showURL{%
\tempurl}


\bibitem[Boehm et~al\mbox{.}(1976)]%
        {Boehm}
\bibfield{author}{\bibinfo{person}{B.~W. Boehm}, \bibinfo{person}{J.~R. Brown},
  {and} \bibinfo{person}{M. Lipow}.} \bibinfo{year}{1976}\natexlab{}.
\newblock \showarticletitle{Quantitative Evaluation of Software Quality}. In
  \bibinfo{booktitle}{\emph{Proceedings of the 2nd International Conference on
  Software Engineering}} (San Francisco, California, USA)
  \emph{(\bibinfo{series}{ICSE '76})}. \bibinfo{publisher}{IEEE Computer
  Society Press}, \bibinfo{address}{Washington, DC, USA},
  \bibinfo{pages}{592–605}.
\newblock


\bibitem[Botella et~al\mbox{.}(2004)]%
        {Botella}
\bibfield{author}{\bibinfo{person}{Pere Botella}, \bibinfo{person}{Xavier
  Burgués}, \bibinfo{person}{Juan-Pablo Carvallo}, \bibinfo{person}{Xavier
  Franch}, \bibinfo{person}{Gemma Grau}, \bibinfo{person}{Jordi Marco}, {and}
  \bibinfo{person}{Carme Quer}.} \bibinfo{year}{2004}\natexlab{}.
\newblock \showarticletitle{ISO/IEC 9126 in practice: what do we need to
  know?}. In \bibinfo{booktitle}{\emph{Software Measurement European Forum
  2004}} (Rome). \bibinfo{pages}{297--306}.
\newblock
\urldef\tempurl%
\url{http://www.lsi.upc.es/~jmarco/publications\_pdfs/SMEF2004.pdf}
\showURL{%
\tempurl}


\bibitem[Breck et~al\mbox{.}(2017)]%
        {Rubric}
\bibfield{author}{\bibinfo{person}{Eric Breck}, \bibinfo{person}{Shanqing Cai},
  \bibinfo{person}{Eric Nielsen}, \bibinfo{person}{Michael Salib}, {and}
  \bibinfo{person}{D. Sculley}.} \bibinfo{year}{2017}\natexlab{}.
\newblock \showarticletitle{The ML Test Score: A Rubric for ML Production
  Readiness and Technical Debt Reduction}. In
  \bibinfo{booktitle}{\emph{Proceedings of IEEE Big Data}}.
\newblock


\bibitem[Cavano and McCall(1978)]%
        {McCall}
\bibfield{author}{\bibinfo{person}{Joseph~P. Cavano} {and}
  \bibinfo{person}{James~A. McCall}.} \bibinfo{year}{1978}\natexlab{}.
\newblock \showarticletitle{A Framework for the Measurement of Software
  Quality}.
\newblock \bibinfo{journal}{\emph{SIGSOFT Softw. Eng. Notes}}
  \bibinfo{volume}{3}, \bibinfo{number}{5} (\bibinfo{date}{jan}
  \bibinfo{year}{1978}), \bibinfo{pages}{133–139}.
\newblock
\showISSN{0163-5948}
\urldef\tempurl%
\url{https://doi.org/10.1145/953579.811113}
\showDOI{\tempurl}


\bibitem[Cohen and Katzir(2008)]%
        {GMC}
\bibfield{author}{\bibinfo{person}{Reuven Cohen} {and} \bibinfo{person}{Liran
  Katzir}.} \bibinfo{year}{2008}\natexlab{}.
\newblock \showarticletitle{The Generalized Maximum Coverage Problem}.
\newblock \bibinfo{journal}{\emph{Inform. Process. Lett.}}
  \bibinfo{volume}{108}, \bibinfo{number}{1} (\bibinfo{year}{2008}),
  \bibinfo{pages}{15--22}.
\newblock
\showISSN{0020-0190}
\urldef\tempurl%
\url{https://doi.org/10.1016/j.ipl.2008.03.017}
\showDOI{\tempurl}


\bibitem[C{\^o}t{\'e} et~al\mbox{.}(2007)]%
        {discipline}
\bibfield{author}{\bibinfo{person}{Marc-Alexis C{\^o}t{\'e}},
  \bibinfo{person}{Witold Suryn}, {and} \bibinfo{person}{Elli Georgiadou}.}
  \bibinfo{year}{2007}\natexlab{}.
\newblock \showarticletitle{In search for a widely applicable and accepted
  software quality model for software quality engineering}.
\newblock \bibinfo{journal}{\emph{Software Quality Journal}}
  \bibinfo{volume}{15} (\bibinfo{year}{2007}), \bibinfo{pages}{401--416}.
\newblock


\bibitem[Dromey(1995)]%
        {Dromey}
\bibfield{author}{\bibinfo{person}{R.~Geoff Dromey}.}
  \bibinfo{year}{1995}\natexlab{}.
\newblock \showarticletitle{A Model for Software Product Quality}.
\newblock \bibinfo{journal}{\emph{IEEE Trans. Softw. Eng.}}
  \bibinfo{volume}{21}, \bibinfo{number}{2} (\bibinfo{date}{feb}
  \bibinfo{year}{1995}), \bibinfo{pages}{146–162}.
\newblock
\showISSN{0098-5589}
\urldef\tempurl%
\url{https://doi.org/10.1109/32.345830}
\showDOI{\tempurl}


\bibitem[Feige(1998)]%
        {feige}
\bibfield{author}{\bibinfo{person}{Uriel Feige}.}
  \bibinfo{year}{1998}\natexlab{}.
\newblock \showarticletitle{A Threshold of Ln n for Approximating Set Cover}.
\newblock \bibinfo{journal}{\emph{J. ACM}} \bibinfo{volume}{45},
  \bibinfo{number}{4} (\bibinfo{date}{jul} \bibinfo{year}{1998}),
  \bibinfo{pages}{634–652}.
\newblock
\showISSN{0004-5411}
\urldef\tempurl%
\url{https://doi.org/10.1145/285055.285059}
\showDOI{\tempurl}


\bibitem[Garey and Johnson(1990)]%
        {NP_Book}
\bibfield{author}{\bibinfo{person}{Michael~R. Garey} {and}
  \bibinfo{person}{David~S. Johnson}.} \bibinfo{year}{1990}\natexlab{}.
\newblock \bibinfo{booktitle}{\emph{Computers and Intractability; A Guide to
  the Theory of NP-Completeness}}.
\newblock \bibinfo{publisher}{W. H. Freeman \& Co.}, \bibinfo{address}{USA}.
\newblock
\showISBNx{0716710455}


\bibitem[Goldengorin(2009)]%
        {NP_Article}
\bibfield{author}{\bibinfo{person}{Boris Goldengorin}.}
  \bibinfo{year}{2009}\natexlab{}.
\newblock \showarticletitle{Maximization of submodular functions: Theory and
  enumeration algorithms}.
\newblock \bibinfo{journal}{\emph{European Journal of Operational Research}}
  \bibinfo{volume}{198} (\bibinfo{date}{10} \bibinfo{year}{2009}),
  \bibinfo{pages}{102--112}.
\newblock
\urldef\tempurl%
\url{https://doi.org/10.1016/j.ejor.2008.08.022}
\showDOI{\tempurl}


\bibitem[Grady(1992)]%
        {Grady}
\bibfield{author}{\bibinfo{person}{Robert~B. Grady}.}
  \bibinfo{year}{1992}\natexlab{}.
\newblock \bibinfo{booktitle}{\emph{Practical Software Metrics for Project
  Management and Process Improvement}}.
\newblock \bibinfo{publisher}{Prentice-Hall, Inc.}, \bibinfo{address}{USA}.
\newblock
\showISBNx{0137203845}


\bibitem[Horkoff(2019)]%
        {Horkoff}
\bibfield{author}{\bibinfo{person}{Jennifer Horkoff}.}
  \bibinfo{year}{2019}\natexlab{}.
\newblock \showarticletitle{Non-Functional Requirements for Machine Learning:
  Challenges and New Directions}. In \bibinfo{booktitle}{\emph{27th {IEEE}
  International Requirements Engineering Conference, {RE} 2019, Jeju Island,
  Korea (South), September 23-27, 2019}},
  \bibfield{editor}{\bibinfo{person}{Daniela~E. Damian}, \bibinfo{person}{Anna
  Perini}, {and} \bibinfo{person}{Seok{-}Won Lee}} (Eds.).
  \bibinfo{publisher}{{IEEE}}, \bibinfo{pages}{386--391}.
\newblock
\urldef\tempurl%
\url{https://doi.org/10.1109/RE.2019.00050}
\showDOI{\tempurl}


\bibitem[{{ISO}/{IEC} 25010}(2011)]%
        {iso25010}
\bibfield{author}{\bibinfo{person}{{{ISO}/{IEC} 25010}}.}
  \bibinfo{year}{2011}\natexlab{}.
\newblock \bibinfo{title}{{ISO}/{IEC} 25010:2011, Systems and software
  engineering — Systems and software Quality Requirements and Evaluation
  (SQuaRE) — System and software quality models}.
\newblock
\newblock


\bibitem[{ISO/IEC 25012}(2008)]%
        {ISO25012}
\bibfield{author}{\bibinfo{person}{{ISO/IEC 25012}}.}
  \bibinfo{year}{2008}\natexlab{}.
\newblock \bibinfo{booktitle}{\emph{Systems and software engineering --
  {S}ystems and software Quality Requirements and Evaluation ({SQuaRE}) --
  {D}ata quality model}}.
\newblock \bibinfo{type}{ISO/IEC} 25012. \bibinfo{institution}{International
  Organization for Standardization}, \bibinfo{address}{Geneva, Switzerland}.
\newblock


\bibitem[{ISO/IEC 9126}(2001)]%
        {ISO9126}
\bibfield{author}{\bibinfo{person}{{ISO/IEC 9126}}.}
  \bibinfo{year}{2001}\natexlab{}.
\newblock \bibinfo{booktitle}{\emph{ISO/IEC 9126. Software engineering --
  Product quality}}.
\newblock \bibinfo{publisher}{ISO/IEC}.
\newblock


\bibitem[Kurakin et~al\mbox{.}(2016)]%
        {Kurakin}
\bibfield{author}{\bibinfo{person}{Alexey Kurakin}, \bibinfo{person}{Ian~J.
  Goodfellow}, {and} \bibinfo{person}{Samy Bengio}.}
  \bibinfo{year}{2016}\natexlab{}.
\newblock \showarticletitle{Adversarial examples in the physical world.}
\newblock \bibinfo{journal}{\emph{CoRR}}  \bibinfo{volume}{abs/1607.02533}
  (\bibinfo{year}{2016}).
\newblock
\urldef\tempurl%
\url{http://dblp.uni-trier.de/db/journals/corr/corr1607.html#KurakinGB16}
\showURL{%
\tempurl}


\bibitem[Lwakatare et~al\mbox{.}(2020)]%
        {LWAKATARE}
\bibfield{author}{\bibinfo{person}{Lucy~Ellen Lwakatare},
  \bibinfo{person}{Aiswarya Raj}, \bibinfo{person}{Ivica Crnkovic},
  \bibinfo{person}{Jan Bosch}, {and} \bibinfo{person}{Helena~Holmström
  Olsson}.} \bibinfo{year}{2020}\natexlab{}.
\newblock \showarticletitle{Large-scale machine learning systems in real-world
  industrial settings: A review of challenges and solutions}.
\newblock \bibinfo{journal}{\emph{Information and Software Technology}}
  \bibinfo{volume}{127} (\bibinfo{year}{2020}), \bibinfo{pages}{106368}.
\newblock
\showISSN{0950-5849}
\urldef\tempurl%
\url{https://doi.org/10.1016/j.infsof.2020.106368}
\showDOI{\tempurl}


\bibitem[Miguel et~al\mbox{.}(2014)]%
        {Miguel}
\bibfield{author}{\bibinfo{person}{Jose Miguel}, \bibinfo{person}{David
  Mauricio}, {and} \bibinfo{person}{Glen Rodriguez}.}
  \bibinfo{year}{2014}\natexlab{}.
\newblock \showarticletitle{A Review of Software Quality Models for the
  Evaluation of Software Products}.
\newblock \bibinfo{journal}{\emph{International journal of Software Engineering
  \& Applications}}  \bibinfo{volume}{5} (\bibinfo{date}{11}
  \bibinfo{year}{2014}), \bibinfo{pages}{31--54}.
\newblock
\urldef\tempurl%
\url{https://doi.org/10.5121/ijsea.2014.5603}
\showDOI{\tempurl}


\bibitem[Myllyaho et~al\mbox{.}(2021)]%
        {MYLLYAHO}
\bibfield{author}{\bibinfo{person}{Lalli Myllyaho}, \bibinfo{person}{Mikko
  Raatikainen}, \bibinfo{person}{Tomi Männistö}, \bibinfo{person}{Tommi
  Mikkonen}, {and} \bibinfo{person}{Jukka~K. Nurminen}.}
  \bibinfo{year}{2021}\natexlab{}.
\newblock \showarticletitle{Systematic literature review of validation methods
  for AI systems}.
\newblock \bibinfo{journal}{\emph{Journal of Systems and Software}}
  \bibinfo{volume}{181} (\bibinfo{year}{2021}), \bibinfo{pages}{111050}.
\newblock
\showISSN{0164-1212}
\urldef\tempurl%
\url{https://doi.org/10.1016/j.jss.2021.111050}
\showDOI{\tempurl}


\bibitem[Nemhauser and Wolsey(1978)]%
        {Nemhauser_best}
\bibfield{author}{\bibinfo{person}{G.L. Nemhauser} {and} \bibinfo{person}{L.A.
  Wolsey}.} \bibinfo{year}{1978}\natexlab{}.
\newblock \bibinfo{booktitle}{\emph{{Best algorithms for approximating the
  maximum of a submodular set function}}}.
\newblock \bibinfo{type}{LIDAM Reprints CORE} 343.
  \bibinfo{institution}{Université catholique de Louvain, Center for
  Operations Research and Econometrics (CORE)}.
\newblock
\urldef\tempurl%
\url{https://doi.org/10.1287/moor.3.3.177}
\showDOI{\tempurl}


\bibitem[Nemhauser et~al\mbox{.}(1978)]%
        {Nemhauser}
\bibfield{author}{\bibinfo{person}{George Nemhauser}, \bibinfo{person}{Laurence
  Wolsey}, {and} \bibinfo{person}{M. Fisher}.} \bibinfo{year}{1978}\natexlab{}.
\newblock \showarticletitle{An Analysis of Approximations for Maximizing
  Submodular Set Functions—I}.
\newblock \bibinfo{journal}{\emph{Mathematical Programming}}
  \bibinfo{volume}{14} (\bibinfo{date}{12} \bibinfo{year}{1978}),
  \bibinfo{pages}{265--294}.
\newblock
\urldef\tempurl%
\url{https://doi.org/10.1007/BF01588971}
\showDOI{\tempurl}


\bibitem[Sculley et~al\mbox{.}(2015)]%
        {HiddenTechDebt}
\bibfield{author}{\bibinfo{person}{D. Sculley}, \bibinfo{person}{Gary Holt},
  \bibinfo{person}{Daniel Golovin}, \bibinfo{person}{Eugene Davydov},
  \bibinfo{person}{Todd Phillips}, \bibinfo{person}{Dietmar Ebner},
  \bibinfo{person}{Vinay Chaudhary}, \bibinfo{person}{Michael Young},
  \bibinfo{person}{Jean-Francois Crespo}, {and} \bibinfo{person}{Dan
  Dennison}.} \bibinfo{year}{2015}\natexlab{}.
\newblock \showarticletitle{Hidden Technical Debt in Machine Learning Systems}.
  In \bibinfo{booktitle}{\emph{Proceedings of the 28th International Conference
  on Neural Information Processing Systems - Volume 2}} (Montreal, Canada)
  \emph{(\bibinfo{series}{NIPS'15})}. \bibinfo{publisher}{MIT Press},
  \bibinfo{address}{Cambridge, MA, USA}, \bibinfo{pages}{2503–2511}.
\newblock


\bibitem[Serban et~al\mbox{.}(2020)]%
        {Serban}
\bibfield{author}{\bibinfo{person}{Alex Serban}, \bibinfo{person}{Koen van~der
  Blom}, \bibinfo{person}{Holger Hoos}, {and} \bibinfo{person}{Joost Visser}.}
  \bibinfo{year}{2020}\natexlab{}.
\newblock \showarticletitle{Adoption and Effects of Software Engineering Best
  Practices in Machine Learning}. In \bibinfo{booktitle}{\emph{Proceedings of
  the 14th ACM / IEEE International Symposium on Empirical Software Engineering
  and Measurement (ESEM)}} (Bari, Italy) \emph{(\bibinfo{series}{ESEM '20})}.
  \bibinfo{publisher}{Association for Computing Machinery},
  \bibinfo{address}{New York, NY, USA}, Article \bibinfo{articleno}{3},
  \bibinfo{numpages}{12}~pages.
\newblock
\showISBNx{9781450375801}
\urldef\tempurl%
\url{https://doi.org/10.1145/3382494.3410681}
\showDOI{\tempurl}


\bibitem[Serban et~al\mbox{.}(2022)]%
        {se_ml_website}
\bibfield{author}{\bibinfo{person}{Alex Serban}, \bibinfo{person}{Koen van~der
  Blom}, \bibinfo{person}{Holger Hoos}, {and} \bibinfo{person}{Joost Visser}.}
  \bibinfo{year}{2022}\natexlab{}.
\newblock \bibinfo{title}{Software engineering for machine learning}.
\newblock
\newblock
\urldef\tempurl%
\url{https://se-ml.github.io/}
\showURL{%
\tempurl}


\bibitem[Siebert et~al\mbox{.}(2021)]%
        {Siebert2021}
\bibfield{author}{\bibinfo{person}{Julien Siebert}, \bibinfo{person}{Lisa
  Joeckel}, \bibinfo{person}{Jens Heidrich}, \bibinfo{person}{Adam Trendowicz},
  \bibinfo{person}{Koji Nakamichi}, \bibinfo{person}{Kyoko Ohashi},
  \bibinfo{person}{Isao Namba}, \bibinfo{person}{Rieko Yamamoto}, {and}
  \bibinfo{person}{Mikio Aoyama}.} \bibinfo{year}{2021}\natexlab{}.
\newblock \showarticletitle{Construction of a quality model for machine
  learning systems}.
\newblock \bibinfo{journal}{\emph{Software Quality Journal}}
  (\bibinfo{year}{2021}).
\newblock


\bibitem[Sviridenko(2004)]%
        {Sviridenko}
\bibfield{author}{\bibinfo{person}{Maxim Sviridenko}.}
  \bibinfo{year}{2004}\natexlab{}.
\newblock \showarticletitle{A Note on Maximizing a Submodular Set Function
  Subject to a Knapsack Constraint}.
\newblock \bibinfo{journal}{\emph{Oper. Res. Lett.}} \bibinfo{volume}{32},
  \bibinfo{number}{1} (\bibinfo{date}{jan} \bibinfo{year}{2004}),
  \bibinfo{pages}{41–43}.
\newblock
\showISSN{0167-6377}
\urldef\tempurl%
\url{https://doi.org/10.1016/S0167-6377(03)00062-2}
\showDOI{\tempurl}


\bibitem[Viera et~al\mbox{.}(2005)]%
        {Viera}
\bibfield{author}{\bibinfo{person}{Anthony~J Viera}, \bibinfo{person}{Joanne~M
  Garrett}, {et~al\mbox{.}}} \bibinfo{year}{2005}\natexlab{}.
\newblock \showarticletitle{Understanding interobserver agreement: the kappa
  statistic}.
\newblock \bibinfo{journal}{\emph{Fam med}} \bibinfo{volume}{37},
  \bibinfo{number}{5} (\bibinfo{year}{2005}), \bibinfo{pages}{360--363}.
\newblock


\bibitem[Vogelsang and Borg(2019)]%
        {Vogelsang}
\bibfield{author}{\bibinfo{person}{Andreas Vogelsang} {and}
  \bibinfo{person}{Markus Borg}.} \bibinfo{year}{2019}\natexlab{}.
\newblock \showarticletitle{Requirements engineering for machine learning:
  Perspectives from data scientists}. In \bibinfo{booktitle}{\emph{2019 IEEE
  27th International Requirements Engineering Conference Workshops (REW)}}.
  IEEE, \bibinfo{pages}{245--251}.
\newblock


\bibitem[Wikipedia(2022)]%
        {List_of_system_quality_attributes}
\bibfield{author}{\bibinfo{person}{Wikipedia}.}
  \bibinfo{year}{2022}\natexlab{}.
\newblock \bibinfo{title}{{List of system quality attributes} ---
  {W}ikipedia{,} The Free Encyclopedia}.
\newblock
  \bibinfo{howpublished}{\url{https://en.wikipedia.org/wiki/List\_of\_system\_quality\_attributes}}.
\newblock
\newblock
\shownote{[Online; accessed 30-April-2022]}.


\bibitem[Wujek and Hall(2016)]%
        {Wujek2016BestPF}
\bibfield{author}{\bibinfo{person}{Brett Wujek} {and} \bibinfo{person}{Patrick
  Hall}.} \bibinfo{year}{2016}\natexlab{}.
\newblock \showarticletitle{Best Practices for Machine Learning Applications}.
\newblock


\bibitem[Zhang et~al\mbox{.}(2020)]%
        {Zhang_ml_testing}
\bibfield{author}{\bibinfo{person}{Jie Zhang}, \bibinfo{person}{Mark Harman},
  \bibinfo{person}{Lei Ma}, {and} \bibinfo{person}{Yang Liu}.}
  \bibinfo{year}{2020}\natexlab{}.
\newblock \showarticletitle{Machine Learning Testing: Survey, Landscapes and
  Horizons}.
\newblock \bibinfo{journal}{\emph{IEEE Transactions on Software Engineering}}
  \bibinfo{volume}{PP} (\bibinfo{date}{02} \bibinfo{year}{2020}),
  \bibinfo{pages}{1--1}.
\newblock
\urldef\tempurl%
\url{https://doi.org/10.1109/TSE.2019.2962027}
\showDOI{\tempurl}


\end{thebibliography}

\appendix

\section{Proof of submodularity of the Coverage function}
\label{appendix:proof}
We prove the submodularity of the coverage function defined in equation \ref{eq:coverage_function}. We will first prove that linear combinations of submodular functions with positive weights are also submodular and then we prove that the coverage function for one single sub-characteristic with weight one is also submodular.

\begin{lemma}
\label{lemma:lin}
Given a finite set $\Omega$, $n \in \mathbb{Z}$, $n+1$ positive monotone submodular functions $f_i$ ($0\leq i \leq n $), $f_i: 2^{\Omega} \rightarrow \mathbb{R}$ then any function $g: 2^{\Omega} \rightarrow \mathbb{R}$ defined as  $g(X) = \sum_i^n w_i f_i(X)$ with $X \in 2^{\Omega}$ and real weights $w_i \geq 0$ then $g$ is submodular.

\end{lemma}

\begin{proof}

$\forall {X,Y \subseteq \Omega}$ and we want to prove $ g(X) + g(Y) \geq g(X \cup Y) + g(X \cap Y)$ (submodularity).

Because $f_i$ is submodular we have, $\forall i \in [0, n]$:

\begin{equation*}
    f_i(X)+f_i(Y) \geq f_i (X \cup Y) + f_i (X \cap Y)
\end{equation*}

Multiply by $w_i \geq 0$:

\begin{equation*}
    w_i f_i(X)+ w_i f_i(Y) \geq w_i f_i (X \cup Y) + w_i f_i (X \cap Y)
\end{equation*}

Then in particular expanding for all $i$: 
\begin{align}
   w_0 f_0(X)+ w_0 f_0(Y) &\geq w_0 f_0 (X \cup Y) + w_0 f_0 (X \cap Y) \notag  \\
   w_1 f_1(X)+ w_1 f_1(Y) &\geq w_1 f_1 (X \cup Y) + w_1 f_1 (X \cap Y)  \notag   \\
   \vdots \notag	\\
   w_n f_n(X)+ w_n f_n(Y) &\geq w_n f_n (X \cup Y) + w_n f_n (X \cap Y) \notag
\end{align}

Summing all inequalities

\begin{small}
\begin{align*}
  \sum_i^n (w_i f_i(X)+ w_i f_i(Y)) &\geq \sum_i^n (w_i f_i (X \cup Y) + w_i f_i (X \cap Y)) \notag
\end{align*}
\begin{equation*}
  \sum_i^n w_i f_i(X)+ \sum_i^n w_i f_i(Y) \geq  
  \sum_i^n w_i f_i (X \cup Y) + \sum_i^n w_i f_i (X \cap Y)           \notag
\end{equation*}
 \end{small}
 Using definition of $g$ we have

\begin{align*}
  g(X) + g(Y) &\geq g(X \cup Y) + g(X \cap Y)
\end{align*}  

This is the definition of submodularity for function $g$ which proves the lemma.

\end{proof}
In the particular case where we have one single sub-characteristic $c$ with weight 1, we can write $f$ as follows:
\begin{equation}
\label{eq:fc}
f(X) = \min(k, \sum_{p \in X}  u(p,c))
\end{equation}
We now focus one proving that the coverage function $f$ is submodular for the specific case of one single sub-characteristic, then, from lemma \ref{lemma:lin} it follows that the general coverage function $f$ as defined in equation  \ref{eq:coverage_function} is also submodular.

\begin{theorem}
\label{theorem:min}

The coverage function $f$ as defined in Equation \ref{eq:coverage_function} (for one sub-characteristic) is submodular.

\end{theorem}

\begin{proof}

$\forall X \subseteq \Omega$ and $a, b \subseteq \Omega \setminus X$ and $a \neq b$ \\

We consider an alternative definition of submodularity: $f$ is submodular $\iff$

\begin{align}
    f(X \cup \{a\}) - f(x) \geq f(X \cup \{a, b\}) - f(X \cup \{b\}) \label{eq:alt_sub_mod}
\end{align}

We define

\begin{align}
    \Delta a &= f(X \cup \{a\}) - f(X) \label{eq:definition_delta_a} \\
    \Delta ba &= f(X \cup \{a, b\}) - f(X \cup \{b\}) \label{eq:definition_delta_ba}
\end{align}
The first equation describes the increment in coverage after adding practice $a$ to set $X$, and the second equation, the increment of coverage after adding again $a$ but to set $X \cup \{b\}$ for some practice $b$. Submodularity captures the diminishing increment in coverage as we expand $X$. Then from Equation \ref{eq:alt_sub_mod} $f$ is submodular $\iff \Delta a \geq \Delta ba$.  \\

We will show that $\Delta a \geq \Delta ba$.

\begin{lemma}
\label{lemma:f_union_x_y}
The coverage function $f$ as defined in Equation \ref{eq:coverage_function} (for one sub-characteristic) satisfies $f(X \cup Y) \geq f(X)$ for any $X,Y \subseteq \Omega$.

\end{lemma}

\begin{proof}
First we rewrite $f$ as follows:
\begin{align*}
\label{eq:s}
&f(X) = \min(k, s(X))
\\
&s(X) = \sum_{p \in X}  u(p,c)
\end{align*}

Now, since the function $s$ is additive and positive, we have:
\begin{align*}
f(X \cup Y) &= min(k, s(X \cup Y)) \\ 
             &= min(k, s(X) + s(Y)) \geq min(k, s(X)) = f(X)
\end{align*}
\end{proof}

We consider the following exhaustive list of cases by splitting on either $f(\cdot) < k$ or $f(\cdot) = k$ (by definition \ref{eq:coverage_function} $f$ cannot be $> k$):

First we consider $f(X) = k$, the coverage function is saturated by set $X$, so adding more practices will not change its value, it follows then that all increments are 0. Formally,
by lemma \ref{lemma:f_union_x_y}  $f(X \cup \{a\}) \geq k$ and by definition of $f$, $f(X \cup \{a\}) \leq k $ from which it follows that $f(X \cup \{a\})=k$, then $\Delta a = 0$. By the same argument, $f(X \cup \{b\})=k$ and $f(X \cup \{a,b\})=k$ from which it follows that $\Delta ba = 0$. Therefore $\Delta a = \Delta ba = 0$.

Next we consider $f(X) < k$ and $f(X \cup \{b\})=k$. Adding a practice $a$ to the set $X$ produces some or no increment of $f$ but, adding the same practice $a$ to $X \cup \{b\}$ produces no increment since $X \cup \{b\}$ saturates $f$, it follows then that $\Delta a \geq \Delta ba$. Formally, by lemma \ref{lemma:f_union_x_y} and the fact that $\forall_{Z \subseteq \Omega} f(Z) \leq k$ , and definition \ref{eq:definition_delta_ba}, $\Delta ba=0$ and since by definition \ref{eq:definition_delta_a} and lemma \ref{lemma:f_union_x_y} $\Delta a \geq 0$ then it follows $\Delta a \geq \Delta ba$.

We now consider the case where $f(X) < k$ and $f(X \cup \{b\}) < k$ and $f(X \cup \{a\}) = k$.  Adding $a$ to $X$ produces an increment that saturates $f$, it follows then that adding $a$ to $X \cup \{b\}$ also saturates $f$, but with a smaller or equal increment due to the contribution of $b$. This means that $\Delta a \geq \Delta ba$. Formally, by lemma \ref{lemma:f_union_x_y} and the fact that $\forall_{Z \subseteq \Omega} f(Z) \leq k$ we have  $f(X \cup \{a, b\}) = k$. Also, by definition \ref{eq:definition_delta_a} $\Delta a = k - f(X)$, and by definition \ref{eq:definition_delta_ba} $\Delta ba = k - f(X \cup \{b\})$ then by lemma \ref{lemma:f_union_x_y} we have $f(X \cup \{b\}) \geq f(X)$, hence $\Delta a \geq \Delta ba$.

Finally we consider the case where $f(X) < k$ and $f(X \cup \{b\}) < k$ and $f(X \cup \{a\}) < k$. Which means that $s(X \cup \{b\}) < k$ and $s(X \cup \{a\}) < k$, then:
\begin{align*}
        \Delta a &= f(X \cup \{a\}) - f(X) \\
                 &= s(X \cup \{a\}) - s(X) \\
                 &= s(X) + s(\{a\}) - s(X) \\
                 &= s(\{a\})
\end{align*}
 and also:
 
\begin{equation}   
\label{ba=f-s}
        \Delta ba = f(X \cup \{a, b\}) - s(X \cup \{b\})
\end{equation}

We now consider two complementary cases: when $f(X \cup \{a, b\}) = k$ and when $f(X \cup \{a, b\}) < k$. \\

In the first case, we have that $s(X \cup \{a, b\}) \geq k$ and then by definition \ref{eq:definition_delta_ba} we have ${\Delta ba = k - s(X \cup \{b\})}$. From this, it follows that:
\begin{align*}
    s(X \cup \{a, b\}) & \geq k \\
    s(X \cup \{b\}) + s(\{a\}) & \geq k \\
    s(\{a\}) & \geq k - s(X \cup \{b\})
\end{align*}

hence $\Delta a \geq \Delta ba.$  \\

Finally, when $f(X \cup \{a, b\}) < k$ we have:

\begin{align*}
    \Delta ba &= s(X \cup \{a, b\}) - s(X \cup \{b\}) \\
              &= s(a) \\
              &= \Delta a
\end{align*}

hence $\Delta a \geq \Delta ba $. 

We have proved that $\Delta a \geq \Delta ba $ for all the cases, it follows then that the function $f$ for one single sub-characteristic as defined in equation \ref{eq:fc} is submodular. Applying lemma \ref{lemma:lin} it follows that the function $f$ is also submodular for the general case of more than one sub-characteristic, which concludes the proof.

\end{proof}

\section{Quality sub-characteristics definitions}
\label{appendix:subchar_definitions}

\subsubsection{Utility}
\begin{description}
    \item \textbf{Accuracy} — The size of the observational error (systematic and random) of an ML system.
    \item \textbf{Effectiveness} — The ability of an ML system to produce a desired result on the business task it is being used for.
    \item \textbf{Responsiveness} — The ability of an ML system to complete the desired task in an acceptable time frame.
    \item \textbf{Reusability} — The ability to reuse the same ML system without any change, for another business case.
\end{description}

\subsection{Economy}
\begin{description}
    \item \textbf{Cost-effectiveness} — The extent to which an ML System achieves the desired relationship between costs and overall impact.
    \item \textbf{Efficiency} — The ability to avoid wasting resources (computational, human, financial, etc.) in order to perform the desired task effectively.
\end{description}

\subsection{Robustness}

\begin{description}
    \item \textbf{Availability} — The probability that the system will be up and running and able to deliver useful services to users.
    \item \textbf{Resilience} — The extent to which an ML system can provide and maintain an acceptable level of service in the face of technical challenges to normal operation.
    \item \textbf{Adaptability} — The extent to which an ML system can adapt to changes in the production environment, always providing the same functioning level.
    \item \textbf{Scalability} — The extent of an ML system to handle a growing amount of work by adding resources to the system.
\end{description}

\subsection{Modifiability}
\begin{description}
    \item \textbf{Extensibility} — The ease with which a system can be modified, in order to be used for another purpose.
    \item \textbf{Maintainability} — The ease with which activities aiming at keeping an ML system functional in the desired regime, can be performed.
    \item \textbf{Modularity} — The extent to which an ML system consists of components of limited functionality that interrelate with each other.
    \item \textbf{Testability} — The extent with which an ML system can be tested against expected behaviours.
\end{description}

\subsection{Productionizability}

\begin{description}
    \item \textbf{Deployability} — The extent to which an ML system can be deployed in production when needed.
    \item \textbf{Repeatability} — The ease with which the ML lifecycle can be repeated.
    \item \textbf{Operability} — The extent to which an ML system can be controlled (disable, revert, upload new version, etc.).
    \item \textbf{Monitoring} — The extent to which relevant indicators of an ML system are effectively observed and integrated in the operation of the system.
\end{description}

\subsection{Comprehensibility}
\begin{description}
    \item \textbf{Discoverability} — The extent to which an ML system can be found (by means of performing a search in a database or any other information retrieval system).
    \item \textbf{Readability} — The ease with which the code of an ML system can be understood.
    \item \textbf{Traceability} — The ability to relate artifacts created during the development of an ML system to how they were generated.
    \item \textbf{Understandability} — The ease with which the implementation and design choices of an ML system can be understood.
    \item \textbf{Usability} — The extent to which an ML system can be effectively used by users.
    \item \textbf{Debuggability} — The extent to which the inner workings of the ML system can be analyzed in order to understand why it behaves the way it does.
\end{description}

\subsection{Responsibility}
\begin{description}

    \item \textbf{Explainability} — The ability to explain the output of an ML system.
    \item \textbf{Fairness} — The extent to which an ML system prevents unjust predictions towards protected attributes (race, gender, income, etc).
    \item \textbf{Ownership} — The extent to which there exists people appointed to maintaining the ML System and supporting all the relevant stakeholders.
    \item \textbf{Standards Compliance} — The extent to which applicable standards are followed in the ML system.
    \item \textbf{Vulnerability} — The ease with which the system can be (ab)used to achieve malicious purposes.
\end{description}

\section{Details on the scoring of the vectors}
\label{appendix:scoring}

\subsection{The scoring instructions}
\label{appendix:scoring_instructions}

The instructions given to the annotators for the scoring are the following: \textit{Based on your experience, please rate the influence of a software engineering best practice on the quality model \textit{sub-characteristics} on the following scale: irrelevant (0), weakly contributes (1), contributes (2), strongly contributes (3) and addresses (4)}.

\section{Practice sets}

We provide the set of 41 internal practices along with their descriptions in \ref{appendix:internal_practices}, the removal of overlapping practices when combining sets in \ref{appendix:overlapping} and a description of the optimal practice set of Table \ref{table:practices_max_quality} in \ref{appendix:explanation_of_best_practices}. The actual external practices can be found in \cite{Rubric}, \cite{Amershi}, \cite{Serban} and \cite{se_ml_website}.

\subsection{Internal Practices}
\label{appendix:internal_practices}
We split the practices in categories based on their topic and we provide a short explanation for each one.

\subsubsection{Reproducibile ML Lifecycle}

\begin{enumerate}
    
    \item \textit{Logging of metadata and artifacts (hyperparams, code version, data used to train, model binary, etc)}

    Logging of all the relevant information which can help to reproduce the ML model.
This can be helpful to reproduce a build model artefact, and also understand what went wrong in case of issues.
Extra thought should be given to versioning the actual runtime environment (see Use of containerized environment)
and any randomness seeds.

    \item \textit{Code versioning}

You make use of software configuration management tools, such as Git, to version your code and be able to go back to
previous versions at any given point in time. If using more volatile code like notebooks, make sure to save the actual
code or the latest committed version and diff.

    \item \textit{Data versioning}
    
You version your data at different stages of processing (i.e. raw, after train/test split, after preprocessing),
automatically. You can retrieve it easily using a unique identifier to a storage system.
You don’t overwrite versioned data and keep it for as long as the model’s lifetime.

    \item \textit{Model versioning}

    Keep track of model versions in production, to understand which model achieved which predictions and evaluation results.

    \item \textit{Write documentation about the ML system}

    You document your ML system and make the documentation discoverable
(using a centralized portal or inside the code itself).
You document both the modeling assumptions and the rationale for them,
e.g. the reasoning for certain decisions (why certain features are bucketized, why some of the features are excluded,
what is the intuition behind a custom loss function).
On top of that, you also document engineering and installation instructions, e.g. all the necessary information that a
practitioner needs to know in order to be able to train, evaluate, deploy, etc. the ML system.

    \item \textit{The ML system has a clear owner}

    Having clear scope around which aspects of the ML system are owned by each individual or team. For example if an
ML system is being used by multiple teams for a variety of use cases, who owns the monitoring of the system’s
performance for each of those use cases?

    \item \textit{Use of containerized environment}

    You use virtualization to create reproducible environments for both training and deployment.
Any change to the environment is written in the configuration file,
and pinned if necessary. This ensures that all the necessary dependencies and settings for both training and deployment
are pre-installed in the environment, so no extra effort is required for this.

\end{enumerate}

\subsubsection{Monitoring \& Alerting}

\begin{enumerate}[resume]
    \item \textit{Model Performance Monitoring}

    Monitor the ML and business related aspects of your model in production, including:

\begin{itemize}
    \item Predictions are created and are similar to the testing distribution
    \item Predictions are correlated with business metrics
    \item Any proxy metric which can be relevant
\end{itemize}

    \item \textit{Monitor Feature Drift}

    You hold invariants of your features as they were distributed when the model was created,
and monitor production data for any drift that can cause your model to under-perform.

    \item \textit{Monitor Feature Parity}

    You monitor parity between online and offline features in real-time and alert when changes occur.

    \item \textit{Alerting}

    Make sure to have alerts (with carefully chosen alerting thresholds) for all the key indicators of an ML system.

\end{enumerate}
\subsubsection{Continuous Integration \& Continuous Delivery}

\begin{enumerate}[resume]
    \item \textit{Automate the ML Lifecycle}
    
    Every time a change is made in an ML system and the model needs to be re-deployed, the whole ML lifecycle needs to be
repeated. If this includes many manual steps, there are chances for errors and hence deploying a problematic system
version. Moreover, running steps manually takes time and leads to resource waste. A way to mitigate these issues is to
automate the whole ML lifecycle, e.g. data collection, feature engineering, training, evaluation, prediction threshold
setting and deployment. Whenever a change is made in the system’s code, the pipeline is triggered and a new version of
the system is automatically deployed into production. Of course there are cases where it is desired to manually inspect
the model before deployment (possible in business critical applications) so in such cases, it is fine to leave the
evaluation and deployment step outside of the automated pipeline.

    \item \textit{Automated tests}

    You test your code and data, both for acceptance and for code correctness.
You use tools like code coverage to validate your code is well-tested. You unit test the system’s components
to make sure each of them functions properly in isolation and you create integration tests for the whole pipeline
to make sure all the components work well together. Every time you make changes to the code, you run the tests to
make sure that the system works properly.

    \item \textit{Fast development feedback loops}

    It is preferable to develop an ML system in small iterations without too many changes happening in each iteration.
This helps with debugging as there are fewer things that can go wrong and it is also simpler to understand what
worked and what did not when AB testing a certain ML system version. In order to have fast development cycles,
it is important that we can make fast changes, for example quickly add features or change models and retrain the system.
\end{enumerate}

\subsubsection{High Quality Features}

\begin{enumerate}[resume]
    \item \textit{Achieve Feature Parity}
    
    You make sure that online feature calculation yields the same results as your offline features.
    \item \textit{Keeping the historical features up to date}
    
    Regularly update historical features that might need pre-computations, such that the model always uses up-to-date
features.
    \item \textit{Remove redundant features}
    
    In Machine Learning, unnecessary complexity is not desirable as it makes the maintainability of the ML systems harder,
without improving their performance.  One typical way to add complexity in an ML system without improving the
performance is by adding redundant features, e.g. features which do not add predictive value. To avoid that,
feature selection procedures should be applied to make sure the system is as simple as possible without compromising
the performance.
    \item \textit{Input Data Validation}

    You apply sanity checks on the data with which the ML model is being trained and evaluated on, to make sure they are
as expected.

    \item \textit{Degenerate Feedback Loops}

    For the cases where the ML system’s features are constantly updated and when those features are influenced by the model
itself, we might end up with wrong experimentation results if we do not do preventive actions.

    \item \textit{Dealing with and preventing the consequences of the dependencies created when using an output of a model as input to another.}

    When a model’s output is being used as input to another model, a small error can be propagated from the source to the
target model and have a large negative impact.
Practices which help to deal with and prevent such negative consequences:
\begin{itemize}
    \item Alerting to the target model that the input of the source model might be broken
    \item If you own the target model, monitor your feature (the one coming from the source model) distribution
    \item Automated retraining and deployment pipeline (in case of an error, fix it quickly)
\end{itemize}

\end{enumerate}
\subsubsection{Ensure long-lasting commercial value}

\begin{enumerate}[resume]
    \item \textit{Peer Code Review}
    
    You request your peers to review your code before going to production. This includes both data processing
and model training code.

    \item \textit{Bot Detection}

    You detect and filter out instances made by bots from training and production data.
    \item \textit{ML evaluation metric correlates with the business metric}

    When developing and deploying an ML system we need to make sure that the offline evaluation metric is a good proxy
for the business metric used to evaluate the model in production. This means that if we observe a certain improvement
in the offline evaluation metric, we should also observe this improvement on the production metric. In this way we are
able to know which of the improvements we apply in the model might be promising, before AB testing them.

    \item \textit{Mimic production setting in offline evaluation}

    The offline evaluation of an ML system should mimic the production environment. For example, if a certain feature
during production has null values, then the same feature should also have null values in the offline test set. If
online assumptions are not imposed during offline evaluation, we risk having an optimistic offline performance.

    \item \textit{Compare with a baseline}

    You compare your model with a simple baseline to ensure ML techniques are actually valuable as well as to benchmark
the ML system.
The simple baseline depends on your model, but it can range anything from a simple heuristic through a
regular expression to a simple ML model.

    \item \textit{Impact of model staleness is known \cite{Rubric}}

    Many production ML systems encounter rapidly changing, non-stationary data. Examples include content recommendation
systems and financial ML applications. For such systems, if the pipeline fails to train and deploy sufficiently
up-to-date models, we say the model is stale. Understanding how model staleness affects the quality of predictions
is necessary to determine how frequently to update the model. If predictions are based on a model trained yesterday
versus last week versus last year, what is the impact on the live metrics of interest? Most models need to be updated
eventually to account for changes in the external world; a careful assessment is important to decide how often to
perform the updates

    \item \textit{Write Modular and Reusable Code}

    You write modular and reusable code, doing your best at breaking up large processes into smaller, cohesive chunks. You enable others to re-use your code fragments by generalizing them and pushing them to a centralized code repository.
To promote reusability of your code, make sure you create your functions in ways that enable extensibility, not just modification; these include:

\begin{itemize}
    \item Program by interface (e.g. array-like) and not by implementation (e.g. using isinstance)
    \item Use the Dependency Injection pattern and follow the Open-Closed Principle
    \item Prefer meaningful arguments (e.g. a callable, data structure) over booleans (e.g. \texttt{use\_new\_schema}) when extending functionality
\end{itemize}

\item \textit{Bias Assessment}

You assess the differences in the prediction error of the ML systems across sensitive variables (such as country) to
detect any bias issues.

\item \textit{Bias Mitigation}

You take actions to mitigate the differences in prediction errors across sensitive variables of interest, to make sure
that they are within acceptable thresholds.

\item \textit{Establish clear success metric before model design}

The success criterion for an ML system is something that directly affects the model design choices taken during
development. For example, the optimal hyper-parameters for a given success metric might be suboptimal for a
different metric and need to be re-tuned. This is why it is important to think about the evaluation metrics
(both offline and online) before the design of the model, and take them into account in all the modeling choices.

\item \textit{Error analysis}

Analyze all the error types in the ML system to understand why it failed and what fixes need to be done to prevent such
errors in the future.

\item \textit{Keep the model up to date}

Make sure that the ML system is trained on data points that are representative of the current production environment.

\item \textit{ML system audit by domain experts}

Domain experts periodically audit the ML system to make sure that it complies with any applicable standards.

\end{enumerate}

\subsubsection{A/B Testing}

\begin{enumerate}[resume]
    \item \textit{A/B test all model changes}
    
    Assess the impact of all the model changes through an A/B test.

    \item \textit{Do not test a model which is not promising in offline evaluation}

    Avoid conducting redundant and costly AB tests when the offline evaluation does not indicate significant differences
with the current production model.

    \item \textit{Understand when A/B assumptions are violated and seek alternative experimentation framework}
    
Always make sure that key AB testing assumptions
(for example the independence assumption)
are met before conducting an experiment. In case they are not met, find alternative experimentation methods.

\end{enumerate}    

\subsubsection{Infrastructure}

\begin{enumerate}[resume]
    \item \textit{Unified environment for all lifecycle steps}

    If all the lifecycle steps of a certain ML system are developed in a unified environment it simplifies code reuse,
shared engineering techniques and patterns, as well as shared code style conventions, development tools, etc.
For instance, if the training and serving environments are developed in different languages, we risk re-implementing
the same functionality (e.g. pre or post-processing) in both environments separately.

    \item \textit{Reuse data, supported internal tooling and solutions}

    Prefer to use org-supported tooling and solutions for your ML needs
(e.g. serving platforms, features stores, model registries).
This enables better sharing of your work with others, makes managing of infrastructure easier and keeps you
updated with any enhancements. Moreover, data re-usage especially in the case of costly annotated labels is
quite beneficial and avoids unnecessary effort and cost.

    \item \textit{Turn experimental code into production code}

    You move your experimental code (usually notebooks) to a runnable, organized and git-committed code when you're done
training and before/while you go to production. Bonus point if you validate that the production code you ported
actually re-creates the model as expected.
In the case of using Jupyter notebook as an exploration environment, there are simple ways to
turn the notebooks
into scripts
and then further re-design the software.

    \item \textit{Shadow Deployment}

    Ability to deploy an ML system without exposing its predictions to the users.
One way this can be useful is to understand the technical performance impact of deploying a model,
e.g. is there any delay on the page load where the model is being used, what is the monetary impact of the delay
to the customer?

    \item \textit{Canary Deployment}

    Ability to gradually roll out an ML system, such that any negative impact due to its deployment is mitigated.
\end{enumerate}

Moreover, after the analysis of the internal practices in Section \ref{sec:analyzing_best_practices}, we added the following practices, to fill in the identified technical gaps:

\begin{itemize}
    \item \textit{Request an ML system security inspection}

    An internal process at Booking.com, where security engineers are reviewing ML systems and highlight potential vulnerabilities. The engineers provide recommendations to the ML system owner on how to address potential issues. 
    \item \textit{Latency and Throughput are measured and requirements are defined}

    In order for ML systems to be responsive, latency and throughput should meet the production criteria. Being able to measure the actual latency and throughput and compare it with the business requirements can guide the ML practitioner on the exact improvements to be made in order for the system to be responsive.
    
    \item \textit{Register the ML system in an accessible registry} 

    Registering an ML system and its related documentation in an accessible  ML registry, can enable its discoverability by interested parties across the organization. 
\end{itemize}

\subsection{Removing Overlapping Practices}
\label{appendix:overlapping}

When combining sets of practices, we remove practices that are overlapping with each other in order to avoid duplicate contributions. The practices removed from the analyses in the Sections \ref{sec:how_many_practices} and \ref{section:which_are_the_best} are the provided below.

Removed from the external sets, when combining all the external sets together in Section \ref{sec:analyzing_best_practices}:

\begin{itemize}
    \item Model Versioning \cite{Amershi}
    \item Data Versioning \cite{Amershi}
    \item Automated Tests \cite{Amershi} 
    \item Model specs are reviewed and submitted. \cite{Rubric}
\end{itemize}

Removed from the internal set (\ref{appendix:internal_practices}), when combining the internal and the external sets in Section \ref{sec:how_many_practices}:

\begin{itemize}
    \item Automate the ML lifecycle
    \item Data Versioning
    \item Code versioning
    \item Model Performance Monitoring
    \item Automated Tests
    \item Remove Redundant Features
    \item Peer Code Reviews
    \item Monitor Feature Parity
    \item Offline Metric Correlates with Production Metric
    \item Impact of model staleness is known
    \item Error Analysis
    \item Input Data Validation
    \item Bias Assessment
    \item Logging of Metadata and Artifacts
    \item Shadow Deployment
    \item ML System Audit by Domain Experts
\end{itemize}

\subsection{Description of the practices which maximize quality}
\label{appendix:explanation_of_best_practices}

We provide a description for the best practices in Table \ref{table:practices_max_quality}.

\begin{itemize}
\item Versioning for Data, Model, Configurations and Scripts: 

Versioning in machine learning involves more components than in traditional software - among the executable code we have to store the training and testing data sets, the configuration files and the final model artifacts. Storing all information allows previous experiments to be reproduced and re-assessed. Moreover, it helps auditing, compliance and backward traceability and compatibility.

\item Continuously Monitor the Behaviour of Deployed Models:

Monitoring plays an important role in production level machine learning. Because the performance between training and production data can vary drastically, it is important to continuously monitor the behaviour of deployed models and raise alerts when unintended behaviour is observed.

\item Unifying and automating ML workflow: 

The pipelines created by ML practitioners are automated, supporting
training, deployment, and integration of models with the
product they are a part of.

\item Remove Redundant Features  (originally named as: No feature's cost is too much):

It is not only a waste of computing resources, but also an ongoing
maintenance burden to include features that add only
minimal predictive benefit.

\item Continuously Measure Model Quality and Performance:

For long running experiments – such as training deep neural networks for object recognition or language tasks – it is important to detect errors as early as possible in the training process. Moreover, it is important to continuously check model quality and performance on different benchmarks which match the production environment as close as possible.

\item All input feature code is tested: 

Feature
creation code may appear simple enough to not need unit
tests, but this code is crucial for correct behavior and so
its continued quality is vital. Bugs in features may be
almost impossible to detect once they have entered the data
generation process, especially if they are represented in both
training and test data.

\item Automate Model Deployment: 

Automated deployment involves automatically packing the model together with its dependencies and ‘shipping’ it to a production server – instead of manually connecting to a server and perform the deployment.

\item Use of Containarized Environment: see Appendix \ref{appendix:internal_practices}

\item Unified Environment for all Lifecycle Steps: see Appendix \ref{appendix:internal_practices}

\item Enable Shadow Deployment: 

Instead of deploying a model straight into production, one can assess its quality and performance using the data from production without allowing the model to make final decisions. This involves deploying the a model to “shadow” or “compete” with the model in production, and redirect the data to both models. The model that is already deployed will still handle all decisions, until the shadow model is assessed and promoted to production.

\item The ML system outperforms a simple baseline (originally named as: A simpler model is not better):

Regularly
testing against a very simple baseline model, such as a linear
model with very few features, is an effective strategy both
for confirming the functionality of the larger pipeline and
for helping to assess the cost to benefit tradeoffs of more
sophisticated techniques.

\item Have Your Application Audited:

In order to gain new insights into your development life-cycle and compare with other applications and regulations, it is recommended to have your application audited by an external, independent and trustable actor. Sharing the outcomes of the audit and a strategy for solving the emergent issues can increase transparency and trust.

\item Monitor model staleness (originally named as: Models are not too stale):

Monitoring how
old the system in production is, using the prior measurement
as a guide for determining what age is problematic enough
to raise an alert.
Surprisingly, infrequently updated models also incur a
maintenance cost. Imagine a model that is manually retrained once or twice a year by a given engineer. If that
engineer leaves the team, this process may be difficult to
replicate – even carefully written instructions may become
stale or incorrect over this kind of time horizon.

\item Use A Collaborative Development Platform:

Collaborative development platforms provide easy access to data, code, information, and tools. They also help teams to keep each other informed, make and record decisions, and work together asynchronously or remotely.

Consistent use of a collaborative development environment implies that all team members make use of the environment for all their tasks and that they follow the same conventions in cases where similar tasks can possibly be carried out in different ways in the environment.

\item Explain Results and Decisions to Users:

Machine learning systems involve highly complex between data, algorithms and models. As a result they are often difficult to understand, even for other experts.

In order to increase transparency and align the application with ethics guidelines, it is imperative to inform users on the reasons why a decision was made. For example, the EU GDPR law, as well as the Credit score in the USA, require the right to an explanation for automated decision making systems.

\item The ML system has a clear owner: see Appendix \ref{appendix:internal_practices}.

\item Assign an Owner to Each Feature and Document its Rational:

In a large data set, with multiple features that are composed from distinct data attributes, it is hard to keep track and understand all features. By assigning an owner and documenting each feature, they become easier to maintain and comprehend.

\item Communicate, Align, and Collaborate With Others: 

The system that your team develops is meant to integrate with other systems within the context of a wider organization. this requires communication, alignment, and collaboration with others outside the team.

\item Perform Risk Assessments:

Unintended negative effects of an machine learning application should not be detected after they have happened in production, nor should we wait for a (rare and expensive) third-party audit to detect them in a late stage of development.

By conducting an internal risk assessment, the expertise of your organisation is leveraged to identify and mitigate negative impacts as early as possible.

\item Peer Review Training Scripts:

Errors and bugs can easily slip in the code during development. In order to enhance code quality, techniques from standard programming – such as peer review – can also be applied to machine learning code.

Peer review is a well known technique, in which members of the team review the code between themselves.

\item Establish Responsible AI Values:
Machine learning applications can severely impact human lives. Avoiding negative impacts, even without malicious intent, requires all stakeholders to operate according to the same ethical values.

\item Document important information about the ML system: see Appendix \ref{appendix:internal_practices}.     
\item Write Modular and Reusable Code: see Appendix \ref{appendix:internal_practices}.       
\item Computing performance has not regressed: 

While measuring computational performance is
a standard part of any monitoring, it is useful to slice
performance metrics not just by the versions and components
of code, but also by data and model versions. Degradations
in computational performance may occur with dramatic
changes (for which comparison to performance of prior
versions or time slices can be helpful for detection) or in
slow leaks (for which a pre-set alerting threshold can be
helpful for detection)
\end{itemize}

\section{Code}
In \ref{appendix:prio_framework} we provide the code with the core functionality for the prioritization framework and in \ref{appendix:analysis_code} the code to reproduce the results mentioned in various sections of our work.

\subsection{The prioritization framework}
\label{appendix:prio_framework}

\begin{minted}{python}

import itertools
from typing import Callable, Set


def get_n_practices_brute_force(
    practice_score_f: Callable[[str, str], int],
    practices_available: Set[str],
    sub_chars: Set[str],
    num_practices: int,
    covered_score: int,
) -> Set[str]:
    """
    :param practice_score_f:
        Function returning the score for a sub characteristic and practice. Needs to return
        a score for each of `sub_chars`x`practices_available`
    :param practices_available:
        Set of all available practices.
    :param sub_chars:
        Set of selected sub characteristic we wish to cover
    :param num_practices:
        Budget of practices we can implement
    :param covered_score:
        The score we consider as covering a given sub characteristic.
    :return:
        A set of practices. Caution: the resulting list is not necessarily a feasible solution covering
        all sub-characteristics. If a feasible solution does not exist, the next best is returned.
    """
    best_solution = None
    best_solution_score = None

    for selected_practices in itertools.combinations(
        practices_available, num_practices
    ):
        sub_char_points = {sub_char: covered_score for sub_char in sub_chars}

        for practice in selected_practices:
            for sub_char in sub_chars:
                sub_char_points[sub_char] -= practice_score_f(sub_char, practice)
                sub_char_points[sub_char] = max(0, sub_char_points[sub_char])

        score_achieved = sum(sub_char_points.values())

        if best_solution_score is None:
            best_solution = selected_practices
            best_solution_score = score_achieved
        else:
            if best_solution_score > score_achieved:
                best_solution_score = score_achieved
                best_solution = selected_practices

    return set(best_solution)


def get_n_practices_greedy(
    practice_score_f: Callable[[str, str], int],
    practices_available: Set[str],
    sub_chars: Set[str],
    num_practices: int,
    covered_score: int,
) -> Set[str]:
    """
    :param practice_score_f:
        Function returning the score for a sub characteristic and practice. Needs to return
        a score for each of `sub_chars`x`practices_available`
    :param practices_available:
        Set of all available practices.
    :param sub_chars:
        Set of selected sub characteristic we wish to cover
    :param num_practices:
        Budget of practices we can implement
    :param covered_score:
        The score we consider as covering a given sub characteristic.
    :return:
        A set of practices. Caution: the resulting list is not necessarily a feasible solution covering
        all sub-characteristics. If a feasible solution does not exist, the next best is returned.
    """
    covered: Set[str] = set()
    selected_practices: Set[str] = set()
    left_to_cover = {sub_char: covered_score for sub_char in sub_chars}
    num_rounds = 0

    while covered != sub_chars and num_rounds < num_practices:
        num_rounds += 1

        scored_practices = []

        for practice in practices_available:
            scaled_score = 0
            for sub_char in sub_chars.difference(covered):
                scaled_score += practice_score_f(sub_char, practice)
            scored_practices.append((practice, scaled_score)),

        # greedy selection by the most r points a practice achieves on all sub characteristics
        selected_practice, _ = sorted(scored_practices, key=lambda x: -x[1])[0]

        # Apply the selected practice to each sub-characteristic
        for sub_char in list(left_to_cover.keys()):
            left_to_cover[sub_char] -= practice_score_f(sub_char, selected_practice)

            # Once a sub characteristic is covered, add it to the set of covered once,
            # we do not need to improve it further
            if left_to_cover[sub_char] <= 0:
                covered.add(sub_char)
                del left_to_cover[sub_char]

        practices_available.remove(selected_practice)
        selected_practices.add(selected_practice)

    return selected_practices

\end{minted}

\subsection{Analysis code}
\label{appendix:analysis_code}

\begin{minted}{python}
import os
from typing import Callable, Dict, List, Set

import matplotlib.pyplot as plt
import numpy as np
import pandas as pd
import seaborn as sns
from scipy.stats import pearsonr
from sklearn.metrics import cohen_kappa_score
from tqdm import notebook
import random

from prioritization_framework import get_n_practices_brute_force, get_n_practices_greedy


def read_public_scores(csv_path: str) -> pd.DataFrame:
    scores = pd.read_csv(csv_path, sep="\t")

    scores.loc[:, "Score Annotator A"] = scores.loc[:, "Score Annotator A"].astype(int)
    scores.loc[:, "Score Annotator B"] = scores.loc[:, "Score Annotator B"].astype(int)
    scores["Score"] = scores[["Score Annotator A", "Score Annotator B"]].mean(axis=1)

    return scores


def read_public_and_private_scores(csv_path: str, normalise: bool = True) -> pd.DataFrame:

    scores = pd.read_csv(csv_path, sep="\t", header=0)

    if normalise:
        return normalize_scores(dataframe_w_scores=scores, max_score=4.0)
    else:
        return scores


def compute_agreement_between_practitioners(
    weights: pd.DataFrame,
    prefix: str,
    agreement_func: callable,
) -> pd.DataFrame:
    data = {}
    cols = sorted(col for col in weights.columns if col.startswith(prefix))
    for col1 in cols:
        name1 = col1[len(prefix) :]
        pers = {}
        for col2 in cols:
            if col1 == col2:
                continue
            name2 = col2[len(prefix) :]
            pers[name2] = agreement_func(weights[col1], weights[col2])
        pers["mean"] = sum(pers.values()) / len(pers)
        data[name1] = pers
    return pd.DataFrame(data)


def kappa_(y1, y2):
    map_unsupported_floats = lambda x: int(x * 4)
    y1 = y1.map(map_unsupported_floats)
    y2 = y2.map(map_unsupported_floats)
    return cohen_kappa_score(y1, y2)


def kappa_large_gap(y1, y2):
    """Compute Cohen's Kappa with only meaningful gaps (a difference of more than 1 step).

    In order to do so correctly, whenever we see a difference of 1 step (0.25 in our vector of range [0,1])
    we randomly choose one of the scores and overwrite the other one, so they would be the same.

    By choosing at random we would approximate the same distribution of scores as the original so the Kappa
    score would still be meaningful (as it takes the score distribution as prior).
    """
    y1, y2 = y1.tolist(), y2.tolist()
    for i, (val1, val2) in enumerate(zip(y1, y2)):
        if abs(val1 - val2) < 0.5:
            new_val = random.choice([val1, val2])
            y1[i], y2[i] = new_val, new_val
    return kappa_(pd.Series(y1), pd.Series(y2))


def normalize_scores(dataframe_w_scores: pd.DataFrame, max_score: float) -> pd.DataFrame:
    dataframe_w_scores[dataframe_w_scores.columns[2:]] = dataframe_w_scores[dataframe_w_scores.columns[2:]].apply(
        lambda x: x / max_score
    )
    return dataframe_w_scores


def create_df_from_csvs(
    csvs: List[str],
    header_row_number: int,
    practice_col: str,
    weights_col: str,
    subchar_col: str,
) -> pd.DataFrame:
    """
    Creates dataframe from a list of csv files. It is important that the name of practice, weights and quality
    sub-characteristic columns are the same among the csvs.
    """
    dataframes = []
    for csv_file in csvs:
        loaded_df = create_dataframe_from_csv(
            csv_file,
            header_row_number=header_row_number,
            practice_col=practice_col,
            subchar_col=subchar_col,
            weights_col=weights_col,
        )
        dataframes.append(loaded_df)

    merged = dataframes[0]
    for dataframe in dataframes[1:]:
        merged = merged.join(dataframe, how="left")

    return merged


def create_dataframe_from_csv(
    vectors: str,
    header_row_number: int,
    practice_col: str,
    weights_col: str,
    subchar_col: str,
) -> pd.DataFrame:
    vectors_df = pd.read_csv(vectors, header=header_row_number)

    csv_name = extract_name_from_csv_path(vectors)
    new_weight_col_name = "weights_" + csv_name
    vectors_df.rename(columns={weights_col: new_weight_col_name}, inplace=True)

    vectors_df.set_index([subchar_col, practice_col], verify_integrity=True, inplace=True)

    return vectors_df[[new_weight_col_name]]


def extract_name_from_csv_path(csv_path: str):
    return os.path.basename(csv_path).partition(".")[0]


def get_cumulative_scales(weights: List[int]) -> List[int]:
    return np.insert(np.cumprod(weights), 0, 0)


def score_sub_chars(
    practice_score_f: Callable[[str, str], int],
    practices_available: Set[str],
    sub_chars: Set[str],
) -> Dict[str, int]:
    score_for_sub_char = {sub_char: 0 for sub_char in sub_chars}
    for sub_char in sub_chars:
        for practice in practices_available:
            score_for_sub_char[sub_char] += practice_score_f(sub_char, practice)
    return score_for_sub_char


def which_are_covered(
    score_per_subchar: Dict[str, int],
    covered_score: int,
) -> Dict[str, bool]:
    return {practice: score >= covered_score for practice, score in score_per_subchar.items()}


def plot_set_contributions(points_for_all_practices: dict, title: str, index: str) -> None:
    scale_cumulative = get_cumulative_scales(weights=[1, 2, 3, 4])

    plt.subplots(figsize=[12, 12])
    sns.set_context("poster")

    points_for_all_practices[index]["Sum of scores"].plot.barh(color="black")
    plt.axvline(
        scale_cumulative[-1],
        color="red",
        ls="--",
        label=f"Covered at {scale_cumulative[-1]} points",
    )
    title = title + "\n"
    plt.title(title)
    plt.xlabel("Weighted contribution score")
    plt.ylabel("")
    sns.despine()

    plt.legend()


def plot_quality_coverage(merged_practices: pd.Series) -> None:
    sub_chars = set(list(merged_practices.index.get_level_values(0).unique()))
    practices = set(merged_practices.index.get_level_values(1))
    covered_score = get_cumulative_scales(weights=[1, 2, 3, 4])[-1]

    practice_score_f = lambda sub_char, practice: merged_practices.loc[(sub_char, practice)]
    xs = []
    ys = []
    for num_practices in range(1, len(practices)+1):
        practices_to_adopt = get_n_practices_greedy(
            practice_score_f=practice_score_f,
            practices_available=practices.copy(),
            sub_chars=sub_chars.copy(),
            num_practices=num_practices,
            covered_score=covered_score,
        )

        scores = score_sub_chars(practice_score_f, practices_to_adopt, sub_chars)
        coverage_results = which_are_covered(scores, covered_score)
        not_covered = [sub_char for sub_char, is_covered in coverage_results.items() if not is_covered]

        xs.append(num_practices)
        ys.append((len(sub_chars) - len(not_covered)) / len(sub_chars))

        if len(not_covered) == 0:
            break

    plt.subplots(figsize=[12, 7.5])
    sns.set_context("poster")

    plt.plot(xs, ys, color="black")
    title = "Quality Coverage"
    plt.title(title)
    plt.xlabel("Number of practices chosen by greedy")
    plt.ylabel("Coverage Fraction")
    sns.despine()


def piecewise_linear_function(x: float) -> float:
    if x > 4 or x < 0:
        raise ValueError(f"Not supported domain of x = {x}")

    if x >= 0 and x <= 2:
        return x
    if x >= 2 and x <= 3:
        return x * 4 - 6
    if x >= 3:
        return x * 18 - 48


piecewise_linear_function = np.vectorize(piecewise_linear_function)


def convert(series: pd.Series) -> pd.Series:
    return piecewise_linear_function(series)


def randomly_permute(x: int, lb: int, up: int) -> int:
    return max(0, min(4, x + np.random.randint(lb, up)))


def sensitivity_analysis(
    N: int, name: str, permutation_f: Callable[[str, str], float], practices: pd.DataFrame
) -> np.array:
    results = []
    for i in notebook.tqdm(range(N), desc=name):
        randomly_permuted_score = practices["Score"].apply(lambda x: permutation_f(x))

        permuted_merged = practices.copy()
        permuted_merged["Score_Scaled"] = convert(randomly_permuted_score)

        sum_permuted = (
        permuted_merged
        .groupby("Subcharacteristics")["Score_Scaled"]
        .agg([np.sum])
        .sort_values(by="sum")
        )

        sum_original = (
        practices
        .groupby("Subcharacteristics")["Score_Scaled"]
        .agg([np.sum])
        .sort_values(by="sum")
    )
        ranked = (
        sum_permuted[["sum"]]
        .join(sum_original[["sum"]], rsuffix="_original")
        .rank()
        )
        results.append(pearsonr(ranked["sum"], ranked["sum_original"])[0])

    results = np.array(results)

    return results

\end{minted}
\newpage
\section{Sample Influence scores}
\label{appendix:u_p_c_scores}
To demonstrate the $u(p,c)$ scores, Table \ref{tab:scores_sample} shows one random sample from non-zero scores for each quality sub-characteristic. The full set of scores is omitted due to its size, but is available upon request.

\begin{table}
    \centering
\begin{tabular}{|l|p{.25\textwidth}|r|}
\hline
& Practice & Score \\
\hline
 accuracy             & Share a Clearly Defined Training Objective within the Team         &     1.5 \\
 \hline
 adaptability         & Run Automated Regression Tests                                     &     2   \\
 \hline
 availability         & Enforce Fairness and Privacy                                       &     0.5 \\
\hline

 cost-effectiveness   & Unifying and automating ML workflow                                &     4   \\
 \hline

 debuggability        & Enable Automatic Roll Backs for Production Models                  &     0.5 \\
 \hline

 deployability        & Provide Audit Trails                                               &     0.5 \\
 \hline

 discoverability      & Work Against a Shared Backlog                                      &     0.5 \\
 \hline

 effectiveness        & Training and serving are not skewed                                &     3   \\
 \hline

 efficiency           & Prediction quality has not regressed.                              &     1   \\
 \hline

 explainability       & Peer Review Training Scripts                                       &     2   \\
 \hline

 extensibility        & documentation                                                      &     2   \\
 \hline

 fairness             & Have Your Application Audited                                      &     2   \\
 \hline

 maintainability      & data versioning                                                    &     2   \\
 \hline

 modularity           & Use Continuous Integration                                         &     1.5 \\
 \hline

 monitoring           & Log Production Predictions with the Model's Version and Input Data &     0.5 \\
 \hline

 operability          & Use A Collaborative Development Platform                           &     2.5 \\
 \hline

 ownership            & Share Status and Outcomes of Experiments Within the Team           &     1.5 \\
 \hline

 readability          & The model is debuggable                                            &     1   \\
 \hline

 repeatability        & Model specs are unit tested.                                       &     3   \\
 \hline

 resilience           & mimic production setting in offline evalutation                    &     1   \\
 \hline

 responsiveness       & Automate Model Deployment                                          &     0.5 \\
 \hline

 reusability          & Automate Configuration of Algorithms or Model Structure            &     0.5 \\
 \hline

 scalability          & Run Automated Regression Tests                                     &     0.5 \\
\hline

 standards compliance & data versioning                                                    &     1   \\
\hline

 testability          & The data pipeline has appropriate privacy controls.                &     1   \\
\hline

 traceability         & Work Against a Shared Backlog                                      &     1.5 \\
\hline

 understandability    & All input feature code is tested.                                  &     3   \\
\hline

 usability            & Communicate, Align, and Collaborate With Others                    &     2.5 \\
\hline

 vulnerability        & Establish Responsible AI Values                                    &     0.5 \\
\hline
\end{tabular}
    \caption{Scores on random practices and sub-characteristics pairs}
    \label{tab:scores_sample}
\end{table}
\end{document}